\documentclass[a4paper,onecolumn,accepted=2022-07-22]{quantumarticle}
\pdfoutput=1
\usepackage[utf8]{inputenc}

\usepackage{amsmath}
\usepackage{accents}
\usepackage{amssymb}
\usepackage{amsthm}
\usepackage{enumerate}
\usepackage{bm}
\usepackage{mathtools}
\usepackage{bbm}
\usepackage{qcircuit}
\usepackage{graphicx}
\usepackage{mathrsfs}
\usepackage[mathcal]{eucal}
\usepackage{verbatim}  %%includes comment environment
\usepackage[left=1in,top=1in,right=1in,bottom=1in]{geometry}
\usepackage[normalem]{ulem}
\usepackage{doi}
\usepackage[numbers]{natbib}
\usepackage{hyperref}
\hypersetup{
    colorlinks=true,
    linkcolor=blue,
    citecolor=blue,
    }
\usepackage[nameinlink]{cleveref} 

\DeclarePairedDelimiter{\ceil}{\lceil}{\rceil}
\graphicspath{ {./images/} }

\newcommand{\bbN}{\mathbb{N}}
\newcommand{\bbR}{\mathbb{R}}

\newcommand{\bbC}{\mathbb{C}}

\newcommand{\ket}[1]{|#1\rangle}
\newcommand{\bra}[1]{\langle#1|}
\newcommand{\ketbra}[2]{\ket{#1}\!\bra{#2}}

\usepackage[dvipsnames,usenames]{xcolor}
\usepackage{complexity}

\newtheorem{theorem}{Theorem}[section]
\newtheorem{proposition}[theorem]{Proposition}
\newtheorem{lemma}[theorem]{Lemma}

\newtheorem{corollary}[theorem]{Corollary}

\theoremstyle{definition}
\newtheorem{definition}[theorem]{Definition}

\newcommand{\UW}{{Department of Physics, University of Washington, Seattle, WA 98195, USA}}
\newcommand{\IQUS}{{InQubator for Quantum Simulation (IQuS), Department of Physics, University of Washington, Seattle, WA 98195, USA}}
\newcommand{\UNITN}{{Dipartimento di Fisica, University of Trento, via Sommarive 14, I–38123, Povo, Trento, Italy}}
\newcommand{\UT}{{Department of Computer Science, University of Toronto, Toronto, ON M5S 2E4, Canada}}
\newcommand{\PNNL}{{Pacific Northwest National Laboratory, Richland, WA 99354, USA}}
\title{\protect\parbox{\textwidth}{\protect\centering Hybridized Methods for Quantum Simulation in the Interaction Picture}}

\author{Abhishek Rajput}
\affiliation{\UW}
\author{Alessandro Roggero}
\affiliation{\IQUS}
\affiliation{\UNITN}
\author{Nathan Wiebe}
\affiliation{\UW}
\affiliation{\UT}
\affiliation{\PNNL}
\date{}
   
\begin{document}

\maketitle

\begin{abstract}
 
     Conventional methods of quantum simulation involve trade-offs that limit their applicability to specific contexts where their use is optimal.  In particular, the interaction picture simulation has been found to provide substantial asymptotic advantages for some Hamiltonians but incurs prohibitive constant factors and is incompatible with methods like qubitization.
     We provide a framework that allows different simulation methods to be hybridized and thereby improve performance for interaction picture simulations over known algorithms. These approaches show asymptotic improvements over the individual methods that comprise them and further make interaction picture simulation methods practical in the near term. Physical applications of these hybridized methods yield a gate complexity scaling as $\log^2 \Lambda$ in the electric cutoff $\Lambda$ for the Schwinger Model and independent of the electron density for collective neutrino oscillations, outperforming the scaling for all current algorithms with these parameters. For the general problem of Hamiltonian simulation subject to dynamical constraints, these methods yield a query complexity independent of the penalty parameter $\lambda$ used to impose an energy cost on time-evolution into an unphysical subspace.
\end{abstract}

\section{Introduction}

Since Feynman's seminal work on the simulation of quantum dynamics with quantum computers~\cite{feynmansimulating}, considerable research has been undertaken on the problem of quantum simulation as it is a major area where quantum computers are expected to outperform classical supercomputers~\cite{lloyd1996universal,aspuru2005simulated,reiher2017elucidating,jordan2012quantum,PhysRevD.101.074038}. The problem of simulation is in effect a compilation problem.  The task in simulation is to generate, for a given Hermitian matrix $H$, evolution time $t$, and error tolerance $\epsilon$ a sequence of quantum gates $U(t)$ such that $\|U(t) - e^{-iHt}\|\le \epsilon$, for an appropriate norm $\|\cdot \|$, and the cost of the sequence of gate operations that comprise $U(t)$ is minimal.  This problem is distinct from ordinary unitary synthesis problems because here we do not explicitly know the matrix elements of $e^{-i Ht}$ and need to construct this unitary only using information about the Hamiltonian $H$.

A variety of simulation methods have been developed to approximate the ideal time-evolution channel. The first, and most space efficient, algorithms are the Trotter-Suzuki formulas and their time-ordered generalizations~\cite{lloyd1996universal,berry2007efficient,wiebe2010higher,Poulin_2011,childs2021theory}, but recent years have seen several additions to the repertoire of quantum simulation techniques. The method of qubitization~\cite{Low_2017,Low_2019,gilyen2019quantum,berry2018improved,poulin2018quantum,martyn2021grand,dong2021} involves the implementation of a walk operator whose eigenvalues are an efficiently computable function of those of $H$ and achieves linear scaling in the simulation time $t$, logarithmic scaling in the inverse error tolerance, and scaling independent of the number of terms in the Hamiltonian.
A major drawback of qubitization is that the method does not apply to time-dependent Hamiltonians. 
Linear combinations of unitaries provides simulation methods~\cite{childs2012hamiltonian,berry2014exponential,kieferova2019simulating,low2018hamiltonian} that address this short coming and allow simulations within the interaction picture at costs that can be exponentially lower than all other known methods~\cite{low2018hamiltonian}; however, these approaches require complicated quantum control logic which can lead to undesirable constant factors~\cite{su2021fault}. 

The quantum stochastic drift protocol~\cite{Campbell_2019}, or qDRIFT, is spiritually related to linear combination of unitaries but uses classically controlled evolutions rather than quantum controlled ones. This approach drifts towards the correct unitary time-evolution with high precision and with a gate complexity independent on the number of terms in the Hamiltonian. qDRIFT was later generalized to the continuous qDRIFT protocol for time-dependent Hamiltonians with an $L^1$-norm scaling in the gate complexity~\cite{berry2020time}. The principal disadvantages of this approach are that it has a larger scaling in the simulation time $t$ compared to other algorithms and does not exploit any commutator structure between the terms of a Hamiltonian.

We develop hybrid algorithms in this paper that combine the various conventional approaches for quantum simulation after moving into the interaction picture (I.P.). This is significant because while the interaction picture simulation method provides the best asymptotic scaling known for many problems, the constant factors involved can make it impractical for many applications~\cite{su2021fault}. We address this by combining algorithms such as qDRIFT and qubitization at different stages of the overall simulation procedure within the interaction picture. Since the interaction picture transformation involves conjugation of Hamiltonian summands $\sum_{k \neq j} H_k$ via $e^{itH_j}$, the unitary invariance of the $L^1$-norm scaling from qDRIFT essentially eliminates the contribution of $H_j$ to the query complexity of the hybrid protocols. A direct application of these methods to physical systems such as the Schwinger Model and collective neutrino oscillations yield improved scaling over current algorithms with respect to certain parameters of interest. The general problem of Hamiltonian simulation constrained to a physical subspace can likewise be efficiently simulated using these algorithms, with a scaling independent of the penalty parameter used to impose an energy cost on projections onto the unphysical subspace. 

{We summarize the scaling of the newly introduced hybrid schemes and compare them to standard approaches in \autoref{tab:scalingA}. These are expressed in terms of the oracle complexity for approximating the time-evolution under a Hamiltonian $H =\sum_{i=1}^L H_i$. For the Trotter/qDRIFT based I.P. methods, we show the asymptotic scaling in terms of queries to oracles $\{W_k\}_{k=1}^L$ implementing $W_k(t)=e^{-iH_kt}$ for any choice of summand $H_k$. For the hybrid qubitization I.P. based methods, the queries are instead to the SELECT/PREPARE oracles (see \Cref{section:qubitization} for details) and the oracle $W_l(t)=e^{-itH_l}$. The latter is specifically used to implement the time evolution of the term $H_l$ to enter the interaction picture, while the oracles $\{W_k\}_{k=1}^L$ above are used to implement all the time-evolutions. The constant $\lambda$ and $\lambda_\alpha$ are obtained by first writing $H$ as a linear combination of unitaries $H=\sum_k \omega_kU_k$ with real $\omega_l>0$. Then we have $\lambda=\sum_k\omega_k$ and $\lambda_\alpha=\sum_{k\neq l}\omega_k=\lambda-\omega_l$. As anticipated above, the hybrid I.P. schemes introduced here can become advantageous when $\lambda_\alpha\ll\lambda$ or $\|H-H_l\|_\infty\ll\|H\|_\infty$, that is, when the Hamiltonian term $H_l$ has a large norm (here and in the rest of the paper, $\|H\|_p$ denotes the Schatten $p$-norm of a matrix. See \Cref{section:notation} for further details).}

\begin{table}[h]
\begin{center}
\begin{tabular}{l|l}
Algorithm & Number of oracle calls to $W_k$ \\&or PREPARE/SELECT and $W_l$ \\ \hline
Trotter~\cite{childs2021theory} & $O \left(\frac{\tilde{\alpha}^{1/p} t^{1 + 1/p}}{\epsilon^{1/p}}\right)$ \\
qDRIFT~\cite{Campbell_2019} & $O\left(\frac{t^2}{\epsilon}\left[\sum_{k=1}^L\|H_k\|_\infty\right]^2\right)$ \\ 
Qubitization~\cite{Low_2017,Low_2019} & $O\left(\lambda t +\frac{\log(1/\epsilon)}{\log(\log(1/\epsilon))}\right)$ \\\\
\hline
Trotter + qDRIFT + I.P. [Cor.~\ref{cor:hybridTrotter}]& $ O\left(\frac{t^2}{\epsilon}\sum_{k\neq l}^L\left[\|H_k\|_\infty^2+\left\|\left[H_k,\sum_{q>k,q\neq l}^L H_q\right]\right\|_\infty\right]\right)$\\
qDRIFT + Qubitization + I.P. [Th.~\ref{thm:IPqDqubitsim}]& $ O\left(\lambda_\alpha t + \left(\frac{\|H- H_l\|_\infty^2 t^2}{\epsilon}\right)\frac{\log(\|H-H_l\|_\infty t/\epsilon)}{\log \log(\|H-H_l\|_\infty t/\epsilon)}\right)$\\ 
\end{tabular}
\caption{Query complexities for standard qDRIFT, Trotter, qubitization, and the hybrid schemes from Corollary~\ref{cor:hybridTrotter} and~\autoref{thm:IPqDqubitsim} where $H = \sum_{i=1}^L H_i$ with $L$ the number of summands in the Hamiltonian $H$, $t$ the simulation time, and $\epsilon$ the simulation error. In the Trotter formula, $p$ is the order of the Trotter formula and $\tilde{\alpha}$ involves sums of commutators nested $p$ times. The query complexity for qDRIFT and Trotter-based algorithms are given in terms of upper bounds for queries to each of the $W_k$ oracles that implement time-evolution under a summand $H_k$. The query complexity for the qubitization methods are given in queries to the oracles $W_l$ implementing time-evolution for the interaction picture transformation, SELECT, and PREPARE. For the latter methods, $H$ is decomposed as a linear combination of unitaries $H=\sum_k \omega_kU_k$ with real $\omega_l>0$. Then $\lambda=\sum_k\omega_k$ and $\lambda_\alpha =\lambda-\omega_l$. The hybrid I.P. schemes can become advantageous when $\lambda_\alpha\ll\lambda$ or $\|H-H_l\|_\infty\ll\|H\|_\infty$.} 
\label{tab:scalingA} 
\end{center}
\end{table}
%
%%\begin{table}[t]
%\begin{center}
%\begin{tabular}{l|l}
%Algorithm & Number of oracle calls to SELECT/PREPARE/$W_l$ \\ \hline
%Qubitization~\cite{Low_2017,Low_2019} & $O\left(\lambda t +\frac{\log(1/\epsilon)}{\log(\log(1/\epsilon))}\right)$ \\ \hline
%qDRIFT + Qubitization + I.P. [Th.~\ref{thm:IPqDqubitsim}]& $ O\left(\lambda_\alpha t + \left(\frac{\|H- H_l\|_\infty^2 t^2}{\epsilon}\right)\frac{\log(\|H-H_l\|_\infty t/\epsilon)}{\log \log(\|H-H_l\|_\infty t/\epsilon)}\right)$
%\end{tabular}
%\caption{Oracle complexities for standard Qubitization~\cite{Low_2017,Low_2019} and the hybrid scheme from Theorem~\ref{thm:IPqDqubitsim}\label{tab:scalingB}}
%\end{center}
%\end{table}

\Cref{section:qsimreview} contains a review of some standard methods of quantum simulation and of the interaction picture. More specifically, in \Cref{section:contqD} we summarize the continuous qDRIFT protocol and the relevant theorems on its query complexity. \Cref{section:qubitization} delves into qubitization and singular value transformations. \Cref{section:trott} contains an overview of Trotterization and a generalization of the first order Trotter-Suzuki formula to time-dependent Hamiltonians. \Cref{section:interactionpic} reviews the interaction picture, the key component of our hybrid protocols. \Cref{section:hybridTrotqD} and \Cref{section:qDRIFTqubithyb} contain the main results on our hybrid protocols with \Cref{section:SchwingerModel}, \Cref{section:neutrinoosc}, and \Cref{section:constraineddynamics} presenting applications of them to the Schwinger Model, collective neutrino oscillations, and constrained Hamiltonian dynamics respectively. The reader can find additional background on the diamond norm in \Cref{section:diamond} and on some of the norm notation used throughout the paper in \Cref{section:notation}.

%Since both of these protocols scale only with the $l^1$ norm of $H$, which is unitarily invariant, the contribution of the interaction term to the overall query complexity is removed. This is particularly useful when dealing with terms that are both fast-forwardable and have large or unbounded norms, as is the case with both the aforementioned systems, since their contribution to the query complexity is sub-dominant compared to the other terms when simulating them. 

%%%%%%%%%%%%%%%%%%%%%%%%%%%%%%%%

%%%%%%%%%%%%%%%%%%%%%%%%%%%%%%%%%%%%%%%%%%%%%%5

\section{Standard Methods of Quantum Simulation}
\label{section:qsimreview}

This section contains brief overviews of interaction picture of quantum mechanics and relevant results from standard methods of quantum simulation such as continuous qDRIFT, qubitization, and Trotterization. Those readers already familiar with these topics can skip to \Cref{section:hybridTrotqD}. 

\subsection{Continuous qDRIFT}
\label{section:contqD}

In this subsection, we outline the continuous qDRIFT protocol used to simulate time-dependent Hamiltonians with a scaling depending only on the $L^1$-norm of the Hamiltonian. At its heart is a classical sampling protocol which randomly samples a simulation time $\tau \in [0,t]$ according to a probability distribution and evolves a given state under the time-independent Hamiltonian $H(\tau)$. The probability distribution is chosen such that it is biased towards $\tau$ with large $||H(\tau)||_{\infty}$. The result is a simulation protocol that stochastically drifts towards the ideal unitary time evolution with small error in the diamond norm. 

We present relevant results from \cite{berry2020time} used throughout this paper without proof. Let $H(\tau)$ be a time dependent Hamiltonian defined for $0 \leq \tau \leq t$. Unless otherwise specified, we make the following assumptions of $H(\tau)$:

\begin{enumerate}
\item It is non-zero and continuously differentiable on $[0,t]$ 
\item It is finite dimensional, i.e. $H: [0,t] \rightarrow \bbC^{M\times M}$
\item There exists an oracle $W \colon \bbR^2 \mapsto \bbC^{M\times M}$ such that for any $\tau \in [0,t]$ and $\Delta \in \bbR$, $W(\tau,\Delta) = e^{-iH(\tau)\Delta}$
\end{enumerate} 

The specific implementation of $W$ depends on the simulation protocol in question. For instance, a concrete realization involves ``qubitization oracles" to be discussed later in the paper. For our present purposes, it suffices to assume the existence of such an oracle and analyze the query complexity of algorithms invoking it as a black box.  

The ideal evolution of $H(\tau)$ for time $t$ is given by $E(t,0) = \exp_{\mathcal{T}}(-i \int_0^t d\tau H(\tau))$ and the quantum channel corresponding to this is
%%%%%%%%%
\begin{equation}
\mathcal{E}(t,0) = E(t,0) \rho E^{\dag}(t,0) = \exp_{\mathcal{T}} \bigg(-i \int_0^t d\tau H(\tau) \bigg) \rho \exp_{\mathcal{T}}^{\dag}\bigg(-i \int_0^t d\tau H(\tau) \bigg)\;, \label{eq:idealChannel} 
\end{equation}
%%%%%
where the subscript $\mathcal{T}$ in $\exp_{\mathcal{T}}$ denotes the time-ordered exponential. Generalizations of these channels to non-zero initial times can be accomplished simply by changing the limits of integration.

Since it is difficult in practice to implement the ideal channel due to the presence of time-ordered exponentials, we can instead approximate it by a mixed unitary channel defined by
%%%%%%
\begin{equation}
\mathcal{U}(t,0)(\rho) = \int_0^t d\tau \ p(\tau) e^{-i\frac{H(\tau)}{p(\tau)}} \rho e^{i\frac{H(\tau)}{p(\tau)}}\;, \label{eq:qdChannel}
\end{equation}
%%%%%%%%%
where $$p(\tau) \coloneqq \frac{||H(\tau)||_{\infty}}{||H||_{\infty, 1}}$$ is a probability density function defined for $0 \leq \tau \leq t$ and $$||H||_{\infty, 1} \coloneqq \int_0^t d \tau \|H\|_{\infty}\;.$$ 

(i.e. the outermost subscript indicates an $L^1$ norm while the innermost subscript indicates a Schatten infinity norm). This channel can be implemented via a classical sampling protocol and has the following features:

\begin{enumerate}[(a)]
    \item $p(\tau)$ is biased towards those $\tau \in [0,t]$ with large $||H(\tau)||_{\infty}$
    \item $p(\tau)$ decreases with the evolution time $t$ since $||H(\tau)||_{\infty,1}$ involves an integral over $[0,t]$ 
    \item With a time $\tau_i \in [0,t]$ obtained from sampling $p(\tau)$, we can query the oracle $W$ cited above by inputting $W(\tau_i,p(\tau_i)^{-1})$ to obtain an implementation of the unitary time-evolution operator $e^{-iH(\tau_i)/p(\tau_i)}$  
\end{enumerate}

This classical sampling protocol and the unitary channel~\eqref{eq:qdChannel} implemented by it is denoted by ``continuous qDRIFT". We assume the spectral norm $||H||_{\infty}$ or an upper bound is already known and that we can efficiently sample from $p(\tau)$. We then have the following theorem when the simulation time $t$ is assumed to be sufficiently small:

\begin{theorem}[$L^1$-norm error bound for continuous qDRIFT, short-time version]
\label{thm:shortsimqD}
Let $H(\tau)$ be a time-dependent Hamiltonian defined for $0 \leq \tau \leq t$ and satisfying conditions $1$ and $2$ above. Define $\mathcal{E}(t,0)$ and $\mathcal{U}(t,0)(\rho)$ as in equations $(1)$ and $(2)$ respectively. Then 
%%%%%%%%%
\begin{equation}
    ||\mathcal{E}(t,0) - \mathcal{U}(t,0)||_{\diamond} \leq 4||H||^2_{\infty,1}\;. \label{eq:L1normshort}
\end{equation}
\end{theorem}

(See \Cref{section:diamond} for information about the diamond norm for quantum channels). When the simulation time $t$ is large, we will need to divide the simulation interval $[0,t]$ into sub-intervals $[t_j,t_{j+1}]$ where $0 = t_0 < t_1 < \cdots < t_r = t$ and apply the continuous qDRIFT protocol within each to control the simulation error. In these cases, we have a ``long-time" version of \autoref{thm:shortsimqD}:

\begin{theorem} ($L^1$-norm error bound for continuous qDRIFT for long simulation time)
\label{thm:longsimqD}
Let $H(\tau)$ be a time-dependent Hamiltonian defined for $0 \leq \tau \leq t$ and satisfying conditions $1$ and $2$ above. Define $\mathcal{E}(t,0)$ and $\mathcal{U}(t,0)(\rho)$ as in $(1)$ and $(2)$ respectively. For any positive integer $r$, there exists a division $0 = t_0 < t_1 < \cdots < t_r = t$
%%%%%%%%
\begin{equation}
    \left \| \mathcal{E}(t,0) - \prod_{j=0}^{r-1} \mathcal{U}(t_{j+1},t_j) \right \|_{\diamond} \leq 4\frac{\|H\|^2_{\infty,1}}{r} \;.\label{eq:L1normlong}
\end{equation}
%%%%%%%%
To ensure the simulation error is at most $\epsilon$, it suffices to choose $$r \geq 4 \ceil*{\frac{\|H\|^2_{\infty,1}}{\epsilon}}\;.$$ 
\end{theorem}

The value of $r$ above can also be interpreted as the query complexity of the continuous qDRIFT protocol, i.e. the number of queries to the oracle $W$ needed to implemented channel~\eqref{eq:qdChannel} and satisfy~\eqref{eq:L1normlong} with error less than $\epsilon$.

For additional information on the diamond norm and notation used in these results, the reader may consult \autoref{section:diamond} and \autoref{section:notation}.

%%%%%%%%%%%%%%%%%%%

\subsection{Qubitization and Singular Value Transformations}
\label{section:qubitization}

Having considered continuous qDRIFT, we now briefly review the basics of the qubitization simulation protocol which we seek to combine with the former. We will also frame qubitization as an example of the general notion of the block-encoding of non-unitary matrices within larger unitary ones. 

Qubitization is a method of Hamiltonian simulation involving the synthesis of the time-evolution operator $e^{iHt}$, where $H$ is a time-independent Hamiltonian, via the implementation of a walk operator $\mathcal{W}(H)$ whose eigenvalues are an efficiently computable function of those of $H$. Assuming that we have decomposed $H$ as a linear combination of unitary matrices, the desired walk operator can implemented with the so-called SELECT and PREPARE qubitization oracles. The spectrum can then be transformed efficiently using techniques involving singular value transformations which transform the singular values of an operator by a polynomial function~\cite{Low_2017,Low_2019}. 

Block-encoding refers to the embedding of a non-unitary matrix $H$ into a larger unitary $U$, typically as the upper-left block of $U$. Once a block-encoding is achieved, a quantum circuit can be expressed in terms of $U$. This greatly broadens the applicability of quantum computers, particularly in the domain of the simulation of unitary quantum dynamics. We largely follow the treatments in \cite{dong2021,Babbush_2018}.

Let $H \in \text{End}(\bbC^N)$, where $N = 2^n$, be a Hermitian operator. Suppose there exists an $(m+n)$-qubit unitary matrix $U_H \in \text{End}(\bbC^{MN})$, where $M = 2^m$, such that $$U_H = \begin{pmatrix}
H/\alpha & \cdot \\
\cdot & \cdot
\end{pmatrix}\;,$$
where $\alpha>0$ is a known normalization constant.
We may then get access to $H/\alpha$ by $$H = (\bra{0}^m \otimes I_n) U_H (\ket{0}^m \otimes I_n)\;.$$ 

To quantify how ``close" the block encoded matrix is to the original one, we introduce the following general definition. This definition can also be extended to the case of block-encodings within superoperators, which we will need to consider for the proofs of some later theorems: 

\begin{definition}[Block Encoding]
\label{def:blockencode}
    Suppose that $A$ is an $n$-qubit operator, $\alpha, \varepsilon \in \bbR_{+}$, and $m \in \bbN$. We then say that the $(m+n)$-qubit unitary $U_H$ is a $(\alpha, m, \varepsilon)$-block-encoding of $A$ if $$\|A - \alpha(\bra{S} \otimes I_n)U_H(\ket{S} \otimes I_n) \|_{\infty} \leq \varepsilon\;.$$ where $\ket{S}$ is an $m$-qubit state.
    
    Similarly, we say that a quantum channel $\Lambda$ is a $(\alpha, m, \varepsilon)$-block-encoding of $A$ if \[\max_{\rho} \|A \rho A^{\dag} - \alpha(\bra{T} \otimes I_n)\Lambda(\ket{T}\bra{T} \otimes \rho)(\ket{T} \otimes I_n)\|_{\infty} \leq \epsilon,\] where the maximization is over density matrices $\rho$ and $\ket{T}$ is an $m$-qubit state. 
\end{definition}

Here $\ket{S}$ or $\ket{T}$ are referred to as the ``signal state". The previous example involving $H$ is a $(1,m,0)$-encoding where $\ket{S} = \ket{0}^m$. 

Now suppose we are given a time-independent $H$. $H$ can be decomposed into a linear combination of unitary operators 
%%%%%%%
\begin{equation}
   H = \sum_{l=0}^{L-1} w_l H_l, \ \ \ w_l \in \bbR^+_0, \ \ \ H_l^2 = I \;. \label{eq:hamiltLCU}
\end{equation}  

Here we assume that any complex phases are absorbed into $U_l$. The two oracles used are a preparation oracle whose action on $\ket{0}^{\log L}$ is defined as follows:
%%%%%%
\begin{equation}
    \text{PREPARE}\;\ket{0}^{\log L} = \sum_{l=0}^{L-1} \sqrt{\frac{w_l}{\lambda}}\ket{l} = \ket{\mathcal{L}}\;,\label{eq:prepDef}
\end{equation} 
%%%%%%
where $$\lambda = \sum_l w_l\;,$$ and a selection oracle whose action on an ancilla register $\ket{l}$ and system register $\ket{\Psi}$ is as follows:
\begin{align} 
    \text{SELECT} = \sum_{l=0}^{L-1} \ket{l}\bra{l} \otimes H_l\;, \label{eq:selDef}\\ 
    \text{SELECT}\;\ket{l}\ket{\Psi} \mapsto \ket{l}H_l \ket{\Psi}\;.\label{eq:selAct}
\end{align}

In other words, the SELECT oracle ``selects" a unitary $H_l$ conditioned on the state of the ancilla register $\ket{l}$. Using~\eqref{eq:selDef} and~\eqref{eq:selAct}, it can be shown that $\text{SELECT}$ squares to the identity operator and can therefore be considered as a ``reflection" operator. Note that we also have the following result for the action of SELECT on $\ket{\mathcal{L}}$:
%%%%%%%%
\begin{equation}
    (\bra{\mathcal{L}} \otimes I)(\text{SELECT})(\ket{\mathcal{L}} \otimes I) = \frac{1}{\lambda} \sum_l w_l H_l = \frac{H}{\lambda}\;. \label{eq:selElem}
\end{equation}

The previous equation is a condition for qubitization and oracles that satisfy this condition are referred to as ``qubitization oracles" \cite{Low_2019}. If we define $$U_H = (\text{PREPARE}^{\dag} \otimes I)(\text{SELECT})(\text{PREPARE} \otimes I)\;,$$ it follows from~\eqref{eq:selElem} that $U_H$ is a $(\|w\|_1, \log L, 0)$-block encoding of $H$, where $\|w\|_1 = \sum_l |w_l|$. 

The desired walk operator, also known as the ``iterate", can now be defined as follows:
%%%%%%%%%
\begin{equation}
    \mathcal{W} = \mathcal{R}_L \cdot \text{SELECT}, \ \ \ \ \mathcal{R}_L = (2\ket{\mathcal{L}}\bra{\mathcal{L}} \otimes I - I)\;.
    \label{eq:walk_op}
\end{equation}
%%%%%%%%%

$\mathcal{W}$ is of the form of a Szegedy walk operator since it is the composition of two reflections. From a lemma by C. Jordan on the common invariant subspaces of two reflections \cite{Jordan1875}, it follows that the Hilbert space of the system decomposes under the action of $\mathcal{W}$ into a direct sum of 1 and 2-dimensional irreducible subspaces, where the latter is spanned by $\ket{\mathcal{L}}\ket{k}$ and an orthogonal state $\ket{\phi_k}$. Here, $\ket{k}$ is an eigenstate of $H$ with eigenvalue $E_k$ and $\ket{\phi_k}$ is the component of $\mathcal{W}\ket{\mathcal{L}}\ket{k}$ orthogonal to $\ket{\mathcal{L}}\ket{k}$. Using~\eqref{eq:selElem}, this can be expressed as 
%%%%%%%%%%
\begin{equation}
    \ket{\phi_k} = \frac{(I - \ket{\mathcal{L}}\bra{\mathcal{L}} \otimes \ket{k}\bra{k}) \cdot \text{SELECT}\ket{\mathcal{L}}\ket{k}}{||(I - \ket{\mathcal{L}}\bra{\mathcal{L}} \otimes \ket{k}\bra{k}) \cdot \text{SELECT}\ket{\mathcal{L}}\ket{k}||} 
    = \frac{(\text{SELECT} - \frac{E_k}{\lambda}I)\ket{\mathcal{L}}{\ket{k}}}{\sqrt{1-(\frac{E_k}{\lambda})^2}}\;. \label{eq:seltopright}
\end{equation}

In the 2-dimensional subspaces, $\mathcal{W}$ acts as a rotation whereas on the 1-dimensional subspaces, it has $\pm 1$ eigenvalues. The matrix elements of $\mathcal{W}$ within a two-dimensional subspace can be computed using the above relations. Using~\eqref{eq:selElem}, the top-left entry is $$\bra{k}\bra{\mathcal{L}}\mathcal{W}\ket{\mathcal{L}}\ket{k} = \frac{E_k}{\lambda}\;,$$ and the upper-right entry using~\eqref{eq:seltopright} is $$\bra{k}\bra{\mathcal{L}}\mathcal{W}\ket{\phi_k} = \sqrt{1 - \bigg(\frac{E_k}{\lambda}\bigg)^2}\;.$$ %%%%%% 

The other elements can be computed in an analogous way and we obtain for the form of the 2-dimensional blocks of $\mathcal{W}$

\begin{equation}
 \begin{bmatrix}
    \frac{E_k}{\lambda} & \sqrt{1-\bigg(\frac{E_k}{\lambda}\bigg)^2} \\
    -\sqrt{1-\bigg(\frac{E_k}{\lambda}\bigg)^2} & \frac{E_k}{\lambda}
  \end{bmatrix}
  = e^{i \arccos(E_k/\lambda)Y}\;. \label{eq:walkOPmatrix}
\end{equation}

The controlled walk operator can be implemented using the circuit in \autoref{fig:walkCirc} \cite{Babbush_2018}. It is clear from this that $\mathcal{W}$ requires one query to SELECT and at most two queries to PREPARE to implement. The controlled-SELECT operation can be approximated as requiring the same gate complexity to implement as the SELECT operation.

\begin{figure}[htp]
    \[
            \Qcircuit @C=1em @R=.5em {
            & \qw & \ctrl{1} & \qw &&&&& \qw  & \ctrl{1} & \ctrl{1} & \qw \\
            \lstick{\ket{\alpha}} & {/} \qw & \multigate{1}{\mathcal{W}} & \qw & = &&& \lstick{\ket{\alpha}} & {/} \qw & \multigate{1}{\text{SELECT}}  & \multigate{1}{\mathcal{R_L}} & \qw & = \\
            \lstick{\ket{\psi}} & {/} \qw & \ghost{\mathbb{W}} & \qw &&&& \lstick{\ket{\psi}} & {/} \qw  & \ghost{\text{SELECT}} & \ghost{\mathcal{R_L}} & \qw 
             }
    \]
    \[
            \Qcircuit @C=1em @R=.5em {
            & \qw & \ctrl{1} & \qw & \gate{Z} & \qw  & \qw \\
            \lstick{\ket{\alpha}} & {/} \qw & \multigate{1}{\text{SELECT}} & \gate{\text{PREPARE}^{\dag}} & \ctrlo{-1}  & \gate{\text{PREPARE}} & \qw \\
            \lstick{\ket{\psi}} & {/} \qw & \ghost{\text{SELECT}} & \qw & \qw  & \qw & \qw 
            }
    \]
\caption{Controlled-walk operator in terms of the SELECT and PREPARE oracles} 
\label{fig:walkCirc} 
\end{figure}
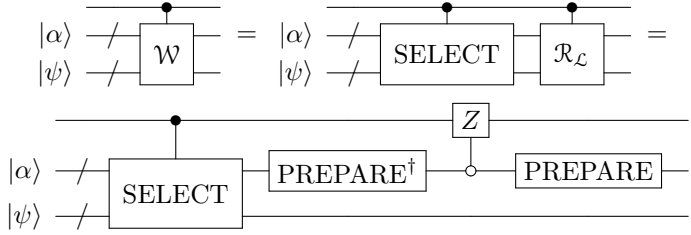

Note that if the condition that $H_l^2 = I$ in~\eqref{eq:hamiltLCU} does not hold, we no longer have the interpretation of SELECT acting like a reflection operator. It then follows that $\mathcal{W}$ cannot be interpreted as a Szegedy walk operator and we can no longer apply Jordan's lemma to it. However, we can still define $\mathcal{W}$ as in~\eqref{eq:walk_op} and the subsequent computations involving the calculation of matrix elements of $\mathcal{W}$ when restricted to the two-dimensional subspace spanned by the orthogonal states $\ket{\mathcal{L}}\ket{k}$ and $\ket{\phi_k}$ remain unaffected. It can still be shown that the Hilbert space decomposes as a direct sum of such 2 dimensional irreducible subspaces as $\mathcal{W}$ does not take vectors within the subspace outside of it.  

The $\arccos$ in~\eqref{eq:walkOPmatrix} can be efficiently inverted to recover the original spectrum of $H$ via techniques involving singular value transformations and quantum signal processing. The impetus for the development of the general formalism of singular value transformations was the Quantum Signal Processing techniques introduced by Low et al. \cite{Low_2016}. They considered the following problem: if one applies a gate sequence of the form $$e^{i \phi_0 \sigma_z} e^{i \theta \sigma_x} e^{i \phi_1 \sigma_z} e^{i \theta \sigma_x} \cdots e^{i \theta \sigma_x} e^{i \phi_k \sigma_z}\;,$$ for unknown $\theta$, where $e^{i \theta \sigma_x}$ is the ``signal unitary" and where we have control over the angles $\phi_0, \cdots,\phi_k$, what unitary operators can be constructed in this manner? This problem lies at the heart of ``Quantum Signal Processing". 
 
The answer to this problem is given in Theorem 3 of \cite{gilyen2019quantum} and involves polynomial transformations of the entries of the signal unitary. This idea behind Quantum Signal Processing can be generalized to situations where we apply an arbitrary unitary $U$ between phase operators. It can be shown that this induces polynomial transformations to the singular values of a particular block of the unitary $U$. In the application to qubitization we are concerned with, Quantum Signal Processing can be applied to the two-dimensional invariant subspaces of the walk operator $\mathcal{W}$. 

As we saw in \Cref{section:qubitization}, qubitization exploits a lemma by C. Jordan's on the invariant subspaces of two reflection operations and the decomposition of the entire vector space into a direct sum of those subspaces. One of the reflections in the lemma can be replaced by a phase gate in the context of quantum search algorithms \cite{Low_2016}. In \cite{gilyen2019quantum}, the other reflection is replaced by an arbitrary unitary $U$ and the invariant subspaces in question are those arising from the singular value decomposition of a block of the unitary matrix. For our purposes, we only need the following results. 

\begin{definition}[Theorem 17 of~\cite{gilyen2019quantum}]
    Let $\mathcal{H}_U$ be a finite-dimensional Hilbert space and $U, \Pi, \tilde{\Pi} \in \text{End}(\mathcal{H}_U)$ be linear operators on $\mathcal{H}_U$ such that $U$ is unitary and $\Pi, \tilde{\Pi}$ are orthogonal projectors. Let $\Phi \in \bbR^n$. Then we define the phased alternating sequence $U_{\Phi}$ as follows 

\[ U_{\Phi} \coloneqq \begin{cases} 
      e^{i \phi_1(2\Pi - I)}U \prod_{j=1}^{(n-1)/2}(e^{i \phi_{2j}(2\Pi - I)} U^{\dag} e^{i \phi_{2j+1}(2\tilde{\Pi}-I)}U) & \text{if n is odd} \\
      \prod_{j=1}^{n/2} (e^{i \phi_{2j-1}(2\Pi - I)}U^{\dag} e^{i \phi_{2j}(2\tilde{\Pi} - I)}U) & \text{if n is even}
   \end{cases}\;.
\]
\end{definition}
\Cref{fig:phasemodcirc} shows the circuit implementation of the alternating phase modulation sequence for even $n$.
%%%%%%%
%\begin{figure}[ht]
%\centering 
%    \[
%           \Qcircuit @C=.6em @R=.5em {
%            & \multigate{2}{U} & \multigate{2}{e^{i\phi_n(2\tilde{\Pi} - I)}} & \multigate{2}{U^{\dag}} & \multigate{2}{e^{i\phi_{n-1}(2\Pi - I)}} & \qw & \multigate{2}{e^{i\phi_3(2 \Pi - I)}} & \multigate{2}{U} & \multigate{2}{e^{i\phi_2(2\tilde{\Pi} -I)}} & \multigate{2}{U^{\dag}} & \multigate{2}{e^{i\phi_1(2\Pi -I)}} & \qw \\
%            & & & & & \cdots & & & & & & & \\
%            & \ghost{U} & \ghost{e^{i\phi_n(2\Pi - I)}} & \ghost{U^{\dag}} & \ghost{e^{i\phi_{n-1}(2\Pi - I)}} & \qw & \ghost{e^{i\phi_3(2 \Pi - I)}} & \ghost{U} & \ghost{e^{i\phi_2(2\Pi -I)}} & \ghost{U^{\dag}} & \ghost{e^{i\phi_1(2\Pi -I)}} & \qw \\
%            }
%    \]
%\caption{Circuit for $U_{\Phi}$ when $n$ is even} 
%\label{fig:phasemodcirc} 
%\end{figure}
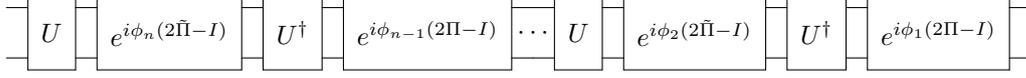
\begin{figure}[ht]
\centering 
    \[
           \Qcircuit @C=.8em @R=.5em {
            & \multigate{2}{U} & \multigate{2}{e^{i\phi_n(2\tilde{\Pi} - I)}} & \multigate{2}{U^{\dag}} & \multigate{2}{e^{i\phi_{n-1}(2\Pi - I)}} & \qw &  \multigate{2}{U} & \multigate{2}{e^{i\phi_2(2\tilde{\Pi} -I)}} & \multigate{2}{U^{\dag}} & \multigate{2}{e^{i\phi_1(2\Pi -I)}} & \qw \\
            & & & & & \cdots & & & & & & & \\
            & \ghost{U} & \ghost{e^{i\phi_n(2\Pi - I)}} & \ghost{U^{\dag}} & \ghost{e^{i\phi_{n-1}(2\Pi - I)}} & \qw  & \ghost{U} & \ghost{e^{i\phi_2(2\Pi -I)}} & \ghost{U^{\dag}} & \ghost{e^{i\phi_1(2\Pi -I)}} & \qw \\
            }
    \]
\caption{Circuit for $U_{\Phi}$ when $n$ is even}
\label{fig:phasemodcirc} 
\end{figure}

\begin{theorem}[Quantum Singular Value Transformation: Theorem 17 of~\cite{gilyen2019quantum}]

Let $\mathcal{H}_U$ be a finite-dimensional Hilbert space and let $U, \Pi, \tilde{\Pi} \in \mathrm{End}(\mathcal{H}_U)$ be linear operators on $\mathcal{H}_U$ such that $U$ is unitary, and $\Pi, \tilde{\Pi}$ are orthogonal projectors. Let $P \in \bbC[x]$ and $\Phi \in \bbR^n$. Then
    
\[ P^{(SV)}(\tilde{\Pi}U\Pi) = \begin{cases} 
      \tilde{\Pi}U_{\Phi}\Pi & \text{if n is odd} \\
      \Pi U_{\Phi} \Pi & \text{if n is even}
   \end{cases}\;,
\]
    where $P^{(SV)}$ is a polynomial of degree at most $n$ that performs a singular value transformation on the operator to which it is applied.
\end{theorem}
%%%%%%
The polynomials in the above theorem are required to satisfy the conditions listed in Corollary 8 of~\cite{gilyen2019quantum}:
\begin{enumerate}[(a)]
    \item $P$ has parity $n \text{ mod } 2$
    \item $\forall x \in [-1,1] \colon |P(x)| \leq 1$
    \item $\forall x \in (-\infty,-1] \cup [1,\infty) \colon |P(x)| \geq 1$
    \item If $n$ is even, then $\forall x \in \bbR \colon P(ix)P^*(ix) \geq 1$
\end{enumerate}

Qubitization works by inverting the arccosine.  While this boils down to the problem of applying a cosine transformation to the input in principle, in practice a Fourier-Chebyshev expansion is used via the Jacobi-Anger expansion that requires both the odd and even terms to closely approximate the desired function and guarantee that the function is within $[-1,1]$ for the entire domain to use the bounds provided in the work.  This process is described in detail in~\cite{Low_2016} as well as in Section 5 of~\cite{gilyen2019quantum}. 

Applying the preceding theorems to $U = \text{CTRL}(\mathcal{W})$ and $\Pi = \tilde{\Pi} = \ket{0}^{L}\bra{0}^{L} \otimes I$, where $L$ is the number of qubits in the register $\ket{{\alpha}}$ in \autoref{fig:walkCirc}, will enable us to invert the $\arccos$ in the spectrum of the walk-operator. A circuit for the unitary operator $e^{i2\phi_j(2\Pi - I)}$ is given in~\Cref{fig:expreflcirc}.

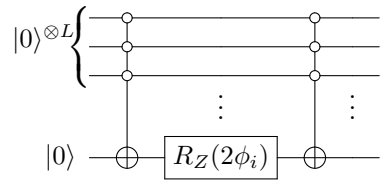
\begin{figure}
\centering 
    \[
           \Qcircuit @C=1em @R=.7em {
            \lstick{} & \ctrlo{1} & \qw & \ctrlo{1} & \qw & \qw \\
            \lstick{} & \ctrlo{1} & \qw & \ctrlo{1} & \qw & \qw \\
            \lstick{} & \ctrlo{3} & \qw & \ctrlo{3} & \qw & \qw
            \inputgroupv{1}{3}{.8em}{.8em}{\ket{0}^{\otimes L}} \\
            & & \vdots & & \vdots & \\
            & \\
            \lstick{\ket{0}} & \targ & \gate{R_Z(2\phi_i)} & \targ & \qw & \qw 
            }
    \]
\caption{Circuit for fractional reflection gadget, $e^{i2\phi_j(2\Pi - I)}$, used in quantum singular value transformations.} 
\label{fig:expreflcirc} 
\end{figure}

Note that we are merely concerned with the existence of a transformation of the spectrum of the walk-operator by a polynomial via the preceding theorem rather than the finding of the phase factors needed to effect a given polynomial transformation. Constructive algorithms for finding these phase factors are outlined in \cite{martyn2021grand,dong2021,chao2020finding,haah2019product}. 

The overall query complexity for qubitization and the singular value transformation is given by the following result as expressed in the language of block encodings. 

\begin{theorem}[Corollary 60 of~\cite{gilyen2019quantum}] 
\label{thm:qubitizquerycost}
    Let $\epsilon \in (0,\frac{1}{2})$, $t \in \bbR$ and $\alpha \in \bbR^{+}$. Let $U$ be an $(\alpha,a,0)$-block encoding of the unknown Hamiltonian $H$. In order to implement an $\epsilon$-precise Hamiltonian simulation unitary $V$ which is an $(1,a+2,\epsilon)$-block encoding of $e^{itH}$, it is necessary and sufficient to use $U$ a total number of times
    
\begin{equation}
    \Theta\bigg(\alpha|t| + \frac{\log(1/\varepsilon)}{\log(e+ \log(1/\varepsilon)/(\alpha|t|))}\bigg)\;. \label{eq:qubitcostgeneral}
\end{equation}

\end{theorem}

Letting $\alpha = \lambda$ in our notation and assuming that $\varepsilon$ is small, we can simplify~\eqref{eq:qubitcostgeneral} as 

\begin{equation}
\Theta\bigg(\lambda t + \frac{\log(1/\varepsilon)}{\log \log(1/\varepsilon)}\bigg)\;. \label{eq:qubitcost}
\end{equation}

The linear term comes from the qubitization portion of the procedure while the logarithmic term stems from the transformation of the singular values via the procedure outlined above. This result can be equivalently interpreted as the query complexity for qubitization in terms of the number of queries (modulo irrelevant constants) needed to the PREPARE and SELECT oracles, since $U = \text{CTRL}(\mathcal{W})$ is related to PREPARE and SELECT via \autoref{fig:walkCirc}.

%%%%%%%%%%%%%%%%%%%%%%%%%

\subsection{Trotterization}

\label{section:trott}

We briefly outline the basics of Trotterization, the oldest method of quantum simulation based on product formulas, and present the relevant results on Trotterization errors used in this paper. Our ultimate goal is to synthesize this method of simulation with the interaction picture and continuous qDRIFT, and compare it with a hybrid continuous qDRIFT and qubitization protocol. This will be followed by an application of both methods to several physical models. 

Let $H = \sum_{i=1}^{\Gamma} H_i$ be a time-\textit{independent} Hamiltonian expressed as a sum of $\Gamma$ terms. The unitary time-evolution operator generated by $H$ is then $e^{it \sum_{i=1}^{\Gamma} H_i}$. There are a variety of product formulas that can be used to decompose the time-evolution operator into a product of exponentials involving the individual terms $H_i$. The most basic is the first-order Lie-Trotter formula $$\mathscr{S}_1(t) \coloneqq e^{itH_{\Gamma}}\cdots e^{it H_1}\;.$$%%%%%%%%%%%% 
%%%%%

Higher order generalizations are the Suzuki formulas defined recursively as $$\mathscr{S}_2(t) \coloneqq e^{i\frac{t}{2} H_1}\cdots e^{i\frac{t}{2}H_{\Gamma-1}} e^{itH_{\Gamma}} e^{i\frac{t}{2}H_{\Gamma-1}} \cdots e^{i\frac{t}{2}H_1}\;,$$
$$\mathscr{S}_{2k}(t) \coloneqq \mathscr{S}_{2k-2}(u_k t)^2 \mathscr{S}_{2k-2}((1-4u_k)t)\mathscr{S}_{2k-2}(u_k t)^2\;,$$
where $u_k = (4 - 4^{-(2k-1)})^{-1}$. There is an extensive literature devoted to investigating the utility and performance of various product-formulas for a variety of physical systems and applications~\cite{reiher2017elucidating,jordan2012quantum,childs2021theory}. While there are multiple strategies for addressing the time-ordering of the operators for the time-ordered operator exponentials that emerge when simulating time-dependent Hamiltonians, we broadly follow the analysis outlined in \cite{Poulin_2011}.

Let $H(t) = \sum_S H_S(t)$ be a time-\textit{dependent} Hamiltonian acting on $N$ particles, where $S \subset \{1,\ldots,N\}$, and each term has bounded norm and acts on at most $k$ particles with $k$ a constant independent of $N$. The time-evolution operator $E(t,0)$ governing the evolution of the system from time $0$ to $t$ is determined by the Schrodinger equation $$\frac{d}{dt} E(t,0) = -iH(t) E(t,0)\;,$$ which admits a solution in terms of a time-ordered exponential $$E(t,0) = \exp_{\mathcal{T}} \bigg \{-i \int_0^t H(s) ds \bigg \}\;.$$

It turns out that the Trotter-Suzuki formulas given above can be generalized to time-dependent scenarios, even in situations where the Hamiltonian experiences fluctuations on time-scales shorter than the time step $\Delta t$ \cite{Poulin_2011}. Suppose we wish to simulate the time-evolution of our system up to time $t_r + \Delta t = T$ from $t_0 = 0$. The exact time-evolution operator can be broken up into shorter segments of the form $$E(T,0) = \prod_{i=0}^r E(t_i + \Delta t, t_i)\;,$$ where
%%%%%
\begin{equation}
    E(t_i + \Delta t,t_i) = \exp_{\mathcal{T}} \bigg(-i \int_{t_j}^{t_j + \Delta t} ds \sum_S H_S(s) \bigg)\;.
\end{equation}
%%%%%

In the case where the sum over $S$ involves only two terms, $H_1$ and $H_2$, the generalized Trotter-Suzuki expansion is of the form 
%%%%
\begin{align}
    E^{\text{TS}}(t_j + \Delta t,t_j) &= \exp_{\mathcal{T}} \bigg(-i \int_{t_j}^{t_j + \Delta t} ds H_1(s) \bigg) \exp_{\mathcal{T}} \bigg(-i \int_{t_j}^{t_j + \Delta t} ds H_2 (s) \bigg) \nonumber \\
    &= E_1^{\text{TS}}(t_j + \Delta t,t_j)E_2^{\text{TS}}(t_j + \Delta t,t_j)\;, \label{eq:gentrotsuz} 
\end{align}
%%%%%
and gives rise to a simulation error of 
%%%%%
\begin{equation}
    \|E(t_j + \Delta t,t_j) - E^{\text{TS}}(t_j + \Delta t,t_j)\|_{\infty} \leq c_{12}(\Delta t^2)\;, \label{eq:gentroterr}
\end{equation} 
%%%%%
where $c_{12}$  is given by
%%%%%
\begin{equation}
    c_{12} = \frac{1}{(\Delta t)^2} \int_{t_j}^{t_j + \Delta t} dv \int_{t_j}^{v} du \| [H_1(u), H_2(v)]\|_{\infty} \leq \frac{1}{2}\max_{u,v}\left(\| [H_1(u), H_2(v)]\|_{\infty} \right)\;. \label{eq:c12def} 
\end{equation}

%%%%%%%%%%%%%%%%%%%%%%%%%%%%%%%

\subsection{The Interaction Picture}
\label{section:interactionpic}

The interaction picture or Dirac picture of quantum mechanics is one of the three representations of operators and states in quantum mechanics \cite{sakurai_napolitano_2017}. It is intermediate to the Schrodinger and Heisenberg pictures of quantum mechanics where the former is characterized by state vectors that evolve in time but with operators constant in time, and vice versa for the latter. Within the interaction picture however, both operators and states have time dependence but the latter evolves according to the interaction Hamiltonian consisting of the left-over terms in the original Hamiltonian. This picture is particularly useful with dealing with terms in a Hamiltonian that can be treated as small perturbations to a main term such as in time-dependent perturbation theory, where it is used in deriving transition rates via Fermi's golden rule and the Dyson series perturbative expansion of the time-evolution operator. It also finds widespread application in interacting quantum field theories. \cite{weinberg_1995}. 

We follow the derivation in \cite{sakurai_napolitano_2017}. Consider a time-independent Hamiltonian $H = \sum_i H_i$. Suppose the energy eigenvalues and eigenstates of $H_j$ for some $j$ are known. 

At $t=t_0$, let the state of the physical system be given by $\ket{\alpha}$. At a later time $t$, we denote the state in the Schrodinger picture by $\ket{\alpha, t_0; t}_S$. Now define 
\begin{equation}
\ket{\alpha, t_0; t}_I \coloneqq e^{iH_j t}\ket{\alpha, t_0; t}_S\;, \label{eq:intket}
\end{equation}
where we have implicitly set $\hbar = 1$ and where the subscript $I$ indicates the same situation as represented in the so-called ``interaction picture" (I.P.). 

We also define observables in the interaction picture as 
%%%%%%%%%
\begin{equation}
A_I(t) \coloneqq e^{i H_j t} A_S e^{-iH_j t}\;. \label{eq:intobs}
\end{equation}

The physical implication of this definition is that we pick any term in the Hamiltonian and move into its ``interaction frame" via conjugation by $e^{iH_j t}$. The major difference between this definition and the analogous one in the Heisenberg picture is the appearance of $H_j$ in the former as opposed to the full $H$ in the latter.

We now take the time derivative of equation~\eqref{eq:intket}: 
\begin{align}
    i \frac{\partial}{\partial t} \ket{\alpha, t_0; t}_I &= i \frac{\partial}{\partial t}\left(e^{iH_j t}\ket{\alpha, t_0; t}_S\right) \nonumber \\
    &= -H_j e^{i H_j t}\ket{\alpha, t_0; t}_S + e^{i H_j t}(H_j + \sum_{i \neq j}H_i)\ket{\alpha, t_0; t}_S \nonumber \\
    &= e^{iH_j t} \sum_{i \neq j} H_i e^{-iH_j t} e^{iH_jt} \ket{\alpha, t_0; t}_S = H_I(t) \ket{\alpha, t_0; t}_I\;,
\end{align}

where we used the Schrodinger equation in the second equality. Thus we have
%%%%%%%
\begin{equation}
i \hbar \frac{\partial}{\partial t} \ket{\alpha, t_0; t}_I = H_I(t) \ket{\alpha, t_0; t}_I\;, \label{eq:intSchro}
\end{equation}
with
%%%%%%%%
\begin{equation}
H_I(t) = e^{iH_j t} \bigg ( \sum_{i\neq j} H_i \bigg ) e^{-iH_j t} = \sum_{i \neq j} H^I_i, \label{eq:intham}
\end{equation}
where $$H^I_i = e^{i H_j t} H_i e^{-i H_j t}\;.$$ This is a Schrodinger-like equation for the time-evolution of the interaction picture state but with the Hamiltonian $H$ replaced by $H_I$. 

It is important to note the distinction between how observables in the interaction picture are represented in~\eqref{eq:intobs} versus the interaction Hamiltonian above. Naively, we would expect the full Hamiltonian to be what is conjugated within the big parentheses in~\eqref{eq:intham} by analogy with~\eqref{eq:intobs}. This ``discrepancy" merely arises from the fact that we needed to define $H_I$ as above to obtain the Schrodinger-like equation~\eqref{eq:intSchro}. 

We can apply the interaction picture to the continuous qDRIFT protocol outlined before and obtain the following simple lemma. 

\begin{lemma} ($L^1$-norm error bound for IP continuous qDRIFT for long simulation time)
\label{lemma:IPlongcontqD}
Let $H_I(\tau)$ be an interaction picture Hamiltonian as in~\eqref{eq:intham}. Suppose it is defined for $0 \leq \tau \leq t$ and satisfies conditions $1$ and $2$ in \Cref{section:contqD}. Define $\mathcal{E}(t,0)$ and $\mathcal{U}(t,0)$ as in~\eqref{eq:idealChannel} and~\eqref{eq:qdChannel}  respectively but with $H_I(\tau)$. Then for any positive integer $r$, there exists a division $0 = t_0 < t_1 < \cdots < t_r = t$ such that 
%%%%%%%%
\begin{equation}
    \left \| \mathcal{E}(t,0) - \prod_{j=0}^{r-1} \mathcal{U}(t_{j+1},t_j) \right \|_{\diamond} \leq 4t^2 \frac{\|\sum_{i \neq j} H_i \|^2_{\infty}}{r}\;. \label{eq:L1normlongIP}
\end{equation}
%%%%%%%%
To ensure the simulation error is at most $\epsilon$, it therefore suffices to choose $$r \geq 4 \ceil*{\frac{t^2 \|\sum_{i \neq j} H_i \|^2_{\infty}}{\epsilon}}\;.$$ 
\end{lemma}

\begin{proof}
We can substitute~\eqref{eq:intham} directly into equation~\eqref{eq:L1normlong} and the expression for $r$. Note however the spectral norm of an operator (and the Schatten norms more generally) is invariant under unitary transformations of that operator. We then obtain the simplification $$||H_I||_{\infty, 1} = \int_0^t d\tau \|H_I(\tau) \|_{\infty} = \bigg \|\sum_{i \neq j} H_i \bigg \|_{\infty} t\;,$$ so that 
%%%%%%%%%%%
\begin{equation}
\bigg \|\mathcal{E}_I(t,0) - \mathcal{U}_I(t,0)\bigg \|_{\diamond} \leq 4\frac{ (\|\sum_{i \neq j} H_i \|_{\infty})^2 t^2}{r}\;,
\end{equation}
and $$r \geq 4 \ceil*{\frac{ (\|\sum_{i \neq j} H_i \|_{\infty})^2 t^2}{\epsilon}}\;,$$ to ensure our simulation error is less than some desired $\epsilon$. \end{proof}

As before, $r$ can also be interpreted as the number of queries to the oracle $W$ defined in \Cref{section:contqD}. Each resulting time-independent piece will need to be simulated using techniques like Trotterization or Qubitization and the main goal of the paper is to quantify the overall query and gate complexity of ``hybrid" protocols combining these with the IP continuous qDRIFT technique outlined here. 

Comparing this result to \autoref{thm:longsimqD}, we see that moving into the interaction frame of a fixed term $H_j$ of the overall Hamiltonian effectively ``eliminates" its contribution to the error. Moreover, due to the properties of the spectral norm and the interaction Hamiltonian, $L^1$-norm dependence of the results in \autoref{thm:longsimqD} reduce to those reminiscent of the time-independent case. This behavior recurs in subsequent results and is particularly useful when dealing with terms with unbounded behavior or large $\infty$-norm, such as the electric term in the Schwinger Model considered later in the paper.

%%%%%%%%%%%%%%%%%%%%%%%%%%%%%%%%%%%%%%%

\section{Hybrid Trotterization and qDRIFT Protocol}
\label{section:hybridTrotqD}

We now present an analysis of our first hybrid simulation protocol where a generalization of the time-dependent Trotter-Suzuki formula given in~\eqref{eq:gentrotsuz} proved below is combined with continuous qDRIFT. Let $H(t) = \sum_{k=1}^L H_k(t)$. The procedure is as follows: 

\begin{enumerate}

\item Use Trotterization technique below to approximate the time-ordered exponential of $H(t)$ as a product of $L$ time-ordered exponentials.

\item Use continuous qDRIFT to approximate each time-ordered exponential by the channel~\eqref{eq:qdChannel}. Implementing this channel involves sampling from a probability distribution and yields a product of $r$ time-\textit{independent} terms of the form $\exp{(-iH_I(\tau_k)/p(\tau_k)})$, where $r$ is the number of sub-intervals of the whole simulation interval. 

\end{enumerate}

Before proving the error bounds for these processes, we first show the following simple lemma with time arguments suppressed for notational convenience:

\begin{lemma}
\label{lemma:supernorms}
Let $\mathcal{E}^{\text{TS}}$ denote the superoperator representing the Trotter-Suzuki decomposition of the time-ordered exponential in~\eqref{eq:gentrotsuz} and let $\mathcal{E}$ be as in~\eqref{eq:idealChannel}. If $D_{2^n}$ is the set of density operators in the domain of $\mathcal{E}$, then
%%%%%%%%%
\begin{equation}
     \|\mathcal{E} - \mathcal{E}^{\text{TS}}\|_{\infty} := \sup_{\rho\in D_{2^n}} \|\mathcal{E}(\rho) - \mathcal{E}^{\text{TS}}(\rho)\|_{\infty} \leq 2\|E - E^{\text{TS}}\|_{\infty}\;.
\end{equation}
\end{lemma} 

\begin{proof}
    From the triangle inequality we have that
    \begin{align*}
         \|\mathcal{E} - \mathcal{E}^{\text{TS}}\|_{\infty} &\leq \sup_{\rho\in D_{2^n}} \|E \rho E^{\dag} - E^{\text{TS}} \rho E^{\dag}\|_{\infty} + \sup_{\rho\in D_{2^n}} \|E^{\text{TS}} \rho E^{\dag} - E^{\text{TS}} \rho (E^{\text{TS}})^{\dag} \|_{\infty} \\
        &= \sup_{\rho\in D_{2^n}} \|(E \rho - E^{\text{TS}} \rho) E^{\dag}\|_{\infty} + \sup_{\rho\in D_{2^n}} \|E^{\text{TS}}(\rho E^{\dag} - \rho (E^{\text{TS}})^{\dag})\|_{\infty} \\
        &= \sup_{\rho\in D_{2^n}} \|(E - E^{\text{TS}}) \rho \|_{\infty} + \sup_{\rho\in D_{2^n}}\|(E - E^{\text{TS}})\rho\|_{\infty} \\
        & \leq 2 \| E - E^{\text{TS}}\|_{\infty}\;.
    \end{align*}
    
    In the third line, we used the unitary invariance of the infinity norm and that $\|A\|_{\infty} = \|A^{\dag}\|_{\infty}$ for any bounded square operator $A$. The latter follows from the fact that the Schatten infinity norm is the spectral norm, which is the largest eigenvalue of $\sqrt{A A^\dagger}$ and coincides with the largest eigenvalue of $A^\dagger A$. In the fourth line, we used the sub-multiplicativity of the infinity norm and the fact that $\|\rho\|_{\infty} \leq 1$ for all density operators.\end{proof}

We now have the following results for quantum simulation with this hybrid protocol: 

\begin{theorem}[Hybrid Trotterization and qDRIFT Simulation]
\label{thm:hybridTrotter}
    Let $\{H_k(t): k=1,\ldots, L\}$ be a set of time-dependent Hermitian operators satisfying conditions 1 and 2 in \Cref{section:contqD}. Let $\mathcal{U}_k$ denote the superoperator representing the continuous qDRIFT channel for the time-dependent summand $H_k(t)$ as in~\eqref{eq:qdChannel}. Then given a decomposition of $[0,t]$ into $r$ sub-intervals of length $\Delta t=t/r$, 
    %%%%%%
    \begin{equation}
        \bigg \|\mathcal{E}(t,0) - \prod_{j=1}^r \prod_{k=1}^L \mathcal{U}_k(t_j + \Delta t,t_j) \bigg \|_{\infty} \leq \frac{L^2 c_{\max}}{r}t^2 + 4r \sum_{k=1}^L \|H_k\|^2_{\infty, 1} \;.\label{eq:trotqDresult} 
    \end{equation}
    Here $c_{\max}$ is defined as $c_{\max} = \frac{1}{L^2}\max_{u,v} \sum_p^L \|[H_p(u), \sum_{q > p}^L H_q(v)]\|_\infty$ and the $1$-norm in $\|H_k\|^2_{\infty,1}$ denotes an integral over an interval of size $\Delta t$.
\end{theorem}

\begin{proof}

We first generalize~\eqref{eq:gentrotsuz} to the case where $H(t)$ is the sum of $L$ time-dependent terms. Suppose we break up $H(t)$ as $H(t) = H_1(t) + \sum_{k>1}^L H_k(t)$. Treating the sum as our ``second" term and considering a specific time-step $[t_l, t_l + \Delta t]$, we can substitute these into the expression for $c_{12}$ above. %However, we instead denote the result by $c_{1,k>1}$ where the $k>1$ subscript indicates that our "second" term consists of summands in the Hamiltonian with index greater than one. We can then break up the "second" term $\sum_{k>1}^L H_k(t) = H_2 + \sum_{k>2}^L H_k$ itself into two terms and repeat the above procedure, giving $c_{2,k>2}$. From the triangle-inequality, we then get  
%%%%%
Our proof of the error bound from recursively applying the bound in~\eqref{eq:gentrotsuz} is inductive.  Let us consider the base case.  Using~\eqref{eq:c12def} we have that
\begin{align}
    &\Biggr\|\exp_\mathcal{T}\left({-i \int_{t_j}^{t_j+\Delta t} H_1(t) + \sum_{k=2}^LH_k(t) \mathrm{d}t}\right) \nonumber\\
    &\qquad- \exp_\mathcal{T}\left({-i \int_{t_j}^{t_j+\Delta t} H_1(t)\mathrm{d}t}\right) \exp_\mathcal{T}\left({-i \int_{t_j}^{t_j+\Delta t} \sum_{k=2}^LH_k(t) \mathrm{d}t}\right) \Biggr\|_\infty\nonumber \\
    &\qquad\le \frac{1}{2}\max_{u,v} \left\|\left[H_1(u), \sum_{q>1}^L H_q(v)\right]\right\|_\infty \Delta t^2 %\le \frac{1}{2}\sum_{q>1}^L \max_{u,v} \left\|\left[H_1(u), H_q(v)\right]\right\| \Delta t^2\;.
    \label{eq:base}
\end{align}
Next assume that for some $p\ge 1$ we have that
%%%%%%%%
\begin{align}
    &\Biggr\|\exp_\mathcal{T}\left({-i \int_{t_j}^{t_j+\Delta t} H_1(t) + \sum_{k=2}^LH_k(t) \mathrm{d}t}\right)\nonumber\\
    &\qquad- \prod_{q=1}^p\exp_\mathcal{T}\left({-i \int_{t_j}^{t_j+\Delta t} H_q(t)\mathrm{d}t}\right) \exp_\mathcal{T}\left({-i \int_{t_j}^{t_j+\Delta t} \sum_{k=p+1}^LH_k(t) \mathrm{d}t}\right) \Biggr\|_\infty\nonumber\\
    &\qquad\le \frac{1}{2}\sum_{\ell=1}^p \max_{u,v} \left\|\left[H_\ell(u), \sum_{q>\ell}^L H_q(v)\right]\right\|_\infty \Delta t^2 \;.
    \label{eq:pformula}
\end{align}
We then have from the triangle inequality and the unitary invariance of Schatten norms that for $p+1$ 
\begin{align}
    &\Biggr\|\exp_\mathcal{T}\left({-i \int_{t_j}^{t_j+\Delta t} H_1(t) + \sum_{k=2}^LH_k(t) \mathrm{d}t}\right) \nonumber\\
    &\qquad- \prod_{q=1}^{p+1}\exp_\mathcal{T}\left({-i \int_{t_j}^{t_j+\Delta t} H_q(t)\mathrm{d}t}\right) \exp_\mathcal{T}\left({-i \int_{t_j}^{t_j+\Delta t} \sum_{{ k=p+2}}^LH_k(t) \mathrm{d}t}\right) \Biggr\|_\infty\nonumber\\
    &\le \frac{1}{2}\sum_{\ell=1}^p \max_{u,v} \left\|\left[H_\ell(u), \sum_{q>\ell}^L H_q(v)\right]\right\| \Delta t^2 + \Biggr\|\exp_\mathcal{T}\left({-i \int_{t_j}^{t_j+\Delta t} \sum_{{k=p+1}}^LH_k(t) \mathrm{d}t}\right)  \nonumber\\
    &\qquad-\exp_\mathcal{T}\left({-i \int_{t_j}^{t_j+\Delta t} H_{{ p+1}}(t)\mathrm{d}t}\right) \exp_\mathcal{T}\left({-i \int_{t_j}^{t_j+\Delta t} \sum_{{ k=p+2}}^LH_k(t) \mathrm{d}t}\right)  \Biggr\|_\infty\nonumber
    \end{align}
    \begin{align}
    &\qquad\le \frac{1}{2}\sum_{\ell=1}^p \max_{u,v} \left\|\left[H_\ell(u), \sum_{q>\ell}^L H_q(v)\right]\right\|_\infty \Delta t^2 + \frac{1}{2} \max_{u,v} \left\|\left[H_{p+1}(u), \sum_{{ q>p+1}}^L H_q(v)\right]\right\|_\infty \Delta t^2\nonumber\\
    &\qquad = \frac{1}{2}\sum_{\ell=1}^{p+1} \max_{u,v} \left\|\left[H_\ell(u), \sum_{q>\ell}^L H_q(v)\right]\right\|_\infty \Delta t^2 \;.
\end{align}
This demonstrates the induction step and combined with the base case in~\eqref{eq:base} shows the error bound we need inductively.

Since this analysis was for the time interval $[t_j + \Delta t, t_j]$ and since there are $r$ such intervals sub-dividing our simulation interval, we can multiply our previous result by $r$ using Box 4.1 in \cite{MikeIke}. Since $\Delta t = t/r$, we then have 
%%%%%
\begin{equation}
    \|E(t,0) - E^{\text{TS}}(t,0)\|_{\infty} \leq \frac{L^2 c_{\max}}{2r}t^2\;.
\end{equation} 
%%%%%
Now note that from~\eqref{eq:L1normshort}  that if we denote the time evolution under $H_k$ to be given by the unitary superoperator $\mathcal{E}_k(t_j+\Delta t,t_j)$, then
%%%%%%%%
\begin{equation}
    \|\mathcal{E}_k(t_j + \Delta t,t_j) - \mathcal{U}_k(t_j + \Delta t,t_j)\|_\infty \le \|\mathcal{E}_k(t_j + \Delta t,t_j) - \mathcal{U}_k(t_j + \Delta t,t_j)\|_{\diamond} \leq 4\|H_k\|^2_{\infty, 1} \;.
\end{equation}
Using the sub-multiplicativity and triangle inequality for the induced infinity norm for superoperators, we get 
%%%%%%%%%
\begin{equation}
    \bigg \|\prod_{k=1}^L\mathcal{E}_k(t_j + \Delta t,t_j) - \prod_{k=1}^L \mathcal{U}_k(t_j + \Delta t,t_j) \bigg \|_{\infty} \leq 4 \sum_{k=1}^L \|H_k\|^2_{\infty, 1}\;,
\end{equation}
%%%%%%%%%
where the $1$-norm in the subscript on the RHS denotes an integral over an interval of size $\Delta t$ from $t_j$ to $t_j + \Delta t$. Note that this notation causes the duration of the integral over time to be implicitly rather than explicitly defined.  Despite this drawback, we use this notation in places throughout the manuscript for brevity.

A straightforward generalization of the argument in Box 4.1 in \cite{MikeIke} using the sub-multiplicativity and triangle inequality for the diamond norm, and the fact quantum channels have diamond norm at most $1$ yields 
%%%%%%%%%%
\begin{equation}
    \bigg \| \prod_{j=1}^r \prod_{k=1}^L\mathcal{E}_k(t_j + \Delta t,t_j) - \prod_{j=1}^r \prod_{k=1}^L \mathcal{U}_k(t_j + \Delta t,t_j) \bigg \|_{\infty} \leq 4r \sum_{k=1}^L \|H_k\|^2_{\infty,1}\;.
\end{equation}
%%%%%%%%%
From the above inequality, \Cref{lemma:IPlongcontqD}, and \Cref{lemma:supernorms} we obtain that the bound of the induced $\infty$-norm of the difference between the super-operator and the hybridized channel is
\begin{align}
    & \bigg \|\mathcal{E}(t,0) - \prod_{j=1}^r \prod_{k=1}^L \mathcal{U}_k(t_j + \Delta t,t_j) \bigg \|_{\infty} \!\! \leq \nonumber \\
    &\bigg \|\mathcal{E}(t,0) - \prod_{j=1}^r\prod_{k=1}^L \mathcal{E}_k(t_j + \Delta t,t_j) \bigg \|_{\infty} \!\!+ \bigg \| \prod_{j=1}^r\prod_{k=1}^L \mathcal{E}_k(t_j + \Delta t,t_j) - \prod_{j=1}^r \prod_{k=1}^L \mathcal{U}_k(t_j + \Delta t,t_j) \bigg \|_{\infty}\nonumber\\ 
    &\leq\frac{L^2c_{\max}}{r}t^2 + 4r \sum_{k=1}^L \|H_k\|^2_{\infty, 1}\;.
\end{align}
%%%%%%%%%%%
\end{proof}

Note that since the 1-norm in $\|H_k\|^2_{\infty, 1}$ denotes an integral over a time-interval of size $\Delta t$, this term scales with $t^2/r$. If we implement each qDRIFT channel $\mathcal{U}_k$ with some error $\epsilon$, an easy application of the triangle inequality will add an additional subdominant term of $rL \epsilon$ to~\eqref{eq:trotqDresult}. 

It should also be noted that the result of~\Cref{thm:hybridTrotter} applies for both the case of time-dependent as well as time-independent Hamiltonian evolution.  This is relevant because it shows that the lowest-order Trotter-Suzuki formula can be combined with qDRIFT profitably wherein small terms in the Hamiltonian can be reallocated between the Trotter and the qDRIFT portions of the Hamiltonian to reduce the simulation cost.  This can be seen as an extension of the coalescing strategy of~\cite{wecker2015solving}.

If we compare this result with that given in Theorem 7' of \cite{berry2020time}, we find that the error in the latter approach using solely continuous qDRIFT scales with $\|H_k\|^2_{\infty, 1,1}$, where the last 1 in the subscript denotes a sum over $k$, and we square \textit{after} performing the integral and sum. While the result in~\eqref{eq:trotqDresult} adds a term which scales at worst quadratically in the number of terms $L$ in the Hamiltonian, we will find that for systems like those considered later in this paper, we can exploit the commutation relations between the terms in the Hamiltonian to give bounds that scale linearly with $L$. 

\begin{corollary}[Hybrid Trotterization and qDRIFT Simulation in Interaction Picture]\label{cor:hybridTrotter} Let $H = \sum_{k=1}^L H_k$ be a time-independent Hamiltonian where each summand satisfies conditions 1 and 2 in \Cref{section:contqD}. Then given a decomposition of $[0,t]$ into $r$ sub-intervals of size $\Delta t$, we can perform the Hamiltonian simulation of $H$ in the interaction frame of $H_l$ as in~\eqref{eq:intham} such that
%%%%%%%%
\begin{equation}
        \bigg \|\mathcal{E}(t,0) - \prod_{j=1}^r \prod_{k \neq l}^L \mathcal{U}_k(t_j + \Delta t, t_j) \bigg \|_{\infty} \leq \frac{t^2}{r} \bigg(c_I + 4 \sum_{k \neq l}^L \|H_k\|^2_{\infty} \bigg)\;, \label{eq:trotqDresultIP}
\end{equation}
where $c_I = \sum_{p \neq l} ^L \|[H_p,\sum_{q>p}^L H_q]\|_{\infty}$. %\ale{[Shouldn't there be an additional factor of $1/2$ here? Also, we should specify which norm and change notation to ensure in the inner sum we are not considering the term $q=l$.]}. 
To ensure the simulation error in the infinity-norm is less than $\epsilon$, it therefore suffices to choose 
%%%%%%%%
\begin{equation}
    r \geq \frac{t^2}{\epsilon} \bigg(c_I + 4 \sum_{k \neq l}^L \|H_k\|^2_{\infty} \bigg)\;. \label{eq:trotqDtstepIP}
\end{equation}

\end{corollary}

\begin{proof}
When moving into the interaction frame of a particular term $H_l$ in $H$ as in~\eqref{eq:intham}, we have $$[H^I_p, H^I_q] = H^I_p H^I_q - H^I_q H^I_p = e^{i H_l t} H_p H_q e^{-iH_l t} - e^{i H_l t} H_q H_p e^{-i H_l t} = [H_p, H_q]^I\;.$$ Since the infinity norm is unitarily invariant, we then have that $$\|[H_p, H_q]^I \|_\infty = \| [H_p, H_q] \|_\infty\;.$$ As the time-dependence came only from the $e^{i H_l t}$ terms, we can drop the maximization over times in $c_{\text{max}}$. The sums in $c_{\text{max}}$ will be over those indices $p,q \neq l$ and we define this simplified quantity as $c_I$ as above.  
    
The $1$-norm in $\|H_k\|^2_{\infty, 1}$ denotes an integral over an interval of measure $\Delta t$, so it again follows from the unitary invariance of the infinity-norm that $\|H_k\|^2_{\infty,1} = (\Delta t)^2 \|H_k\|_{\infty}^2$. Substituting $\Delta t = t/r$ into~\eqref{eq:trotqDresult} then yields the desired expression.
\end{proof}

We can frame the complexity of the preceding process in terms of oracles defined as follows:

\begin{definition}
    Let $H = \sum_k H_k$ be a \textit{time-independent} Hamiltonian in $\bbC^{M \times M}$. We define oracles $\{W_k\}_{k=1}^L$ such that for each $k$, $W_k: \bbR \mapsto \bbC^{M \times M}$ with the action $W_k(\Delta) = e^{-i H_k \Delta}$. 
\end{definition}

These oracles can be used to implement the interaction frame transformation and the time evolution under specific summands of $H$ at various fixed times. Equation~\eqref{eq:trotqDtstepIP} then gives an upper bound on the number of queries to the oracles $W_k$ needed to ensure the simulation protocol is within error $\epsilon$. 

%%%%%%%%%%%%%%%%%%%%%%%%%%%%%%%%%%%%%%%

\section{Hybrid Continuous qDRIFT and Qubitization Protocol}
\label{section:qDRIFTqubithyb} 

We would also like to consider the scenario where we simulate a time-\textit{independent} Hamiltonian $H$ with the following procedure:

\begin{enumerate}

\item Move into the interaction frame of a term $H_j$ in $H$ to turn the simulation problem into one involving a time-\textit{dependent} interaction Hamiltonian $H_I(\tau)$ as in~\eqref{eq:intham}.

\item Use continuous qDRIFT to approximate the ideal time-evolution by the channel~\eqref{eq:qdChannel}. Implementing this channel involves sampling from a probability distribution and yields a product of $r$ time-\textit{independent} terms of the form $\exp{(-iH_I(\tau_k)/p(\tau_k)})$, where $r$ is the number of sub-intervals of the whole simulation interval. 

\item Use qubitization to simulate each time-independent term above and perform a singular value transformation to transform the spectrum in~\eqref{eq:walkOPmatrix} and recover the original spectrum of $H$.

\end{enumerate}

We first make the following definition:

\begin{definition}
\label{def:Woracles}
    Let $H = \sum_k w_k H_k$ be a time-independent Hamiltonian in $\bbC^{M \times M}$. We define an oracle $W_j$ such that $W_j: \bbR \mapsto \bbC^{M \times M}$ with the action $W_j(\Delta) = e^{-i H_j \Delta}$ 
\end{definition}

We use this oracle to transform to the interaction frame of a particular summand $H_j$ in the Hamiltonian $H$ in the following theorem:

\begin{theorem}[Hybrid qDRIFT and Qubitization I.P. Simulation]
\label{thm:IPqDqubitsim}
Let $H=H_j + H_{\alpha} \in \mathbb{C}^{2^n\times 2^n}$ be a time-independent Hamiltonian such that $H_{\alpha}$ has an LCU decomposition $H_{\alpha} = \sum_{l \neq j}^L w_l H_l$, where $w_l \in \mathbb{R}^+$, and each $w_l$ and $H_l$ are obtained by oracles $\mathrm{PREPARE}$ and $\mathrm{SELECT}$ in~\eqref{eq:prepDef} and~\eqref{eq:selDef} respectively.

There exists a quantum algorithm such that for any $\epsilon, t > 0$, it implements a quantum channel $\Lambda$ that is a $(1,O(\log L), \epsilon)$ block-encoding  of $e^{-iHt}$ using a number of queries to $\mathrm{PREPARE}$, $\mathrm{SELECT}$, and $W_j(t)$ in
\begin{equation}
    O \bigg(\lambda_{\alpha} t + \left(\frac{\|H_{\alpha}\|_\infty^2 t^2}{\epsilon}\right)\frac{\log(\|H_{\alpha}\|_\infty t/\epsilon)}{\log \log(\|H_{\alpha}\|_\infty t/\epsilon)}\bigg) \;,\label{eq:mainthm}
\end{equation}
where $\lambda_{\alpha} = \sum_{l \neq j} |w_l|$. 

%\ale{[The queries to $U_j$ should be only $r=O\left(\frac{\|H_{\alpha}\|_\infty^2}{\epsilon}\right)$. The same applies for Corollaries \ref{corollary_SM1} and \ref{corollary_SM2}.]}

\end{theorem}

\begin{proof}
From Theorem 2.2, we have that for any positive integer $r$, there exists a division of $[0,t]$ where $0 = t_0 < t_1 < \cdots < t_k < \cdots < t_r = t$ such that~\eqref{eq:L1normlong} holds, where $\mathcal{E}(t,0)$ and each $\mathcal{U}(t_k,t_{k+1})$ are understood as involving the interaction Hamiltonian $H_I(\tau)$ of~\eqref{eq:intham}. 

By equation~\eqref{eq:L1normlong} $$\left \| \mathcal{E}(t,0) - \prod_{j=0}^{r-1} \mathcal{U}(t_{j+1},t_j) \right \|_{\diamond} \leq 4\frac{\|H_I\|^2_{\infty,1}}{r}\;.$$ %%%%

From the relationship of the trace norm to the diamond norm in~\eqref{eq:tracediamond} and the monotonicity of the Schatten $p$-norm, we get after choosing $r \geq 8 \frac{\|H_I\|^2_{\infty,1}}{\epsilon}$ and defining $D_{2^n}$ to be the set of all density operators in $\mathbb{C}^{2^n\times 2^n}$
%%%%%%%%%%
\begin{align}
&\left \| \mathcal{E}(t,0) - \prod_{j=0}^{r-1} \mathcal{U}(t_{j+1},t_j) \right \|_{\infty}:= \max_{\rho \in D_{2^n}}\left \| \mathcal{E}(t,0)\circ\rho - \left(\prod_{j=0}^{r-1} \mathcal{U}(t_{j+1},t_j)\right)\circ\rho \right \|_{\infty} \nonumber\\
&\leq \max_{\rho \in D_{2^n}}\left \| \mathcal{E}(t,0)\circ\rho - \left(\prod_{j=0}^{r-1} \mathcal{U}(t_{j+1},t_j)\right)\circ\rho \right \|_{1} =\left \| \mathcal{E}(t,0) - \prod_{j=0}^{r-1} \mathcal{U}(t_{j+1},t_j) \right \|_{\diamond} \nonumber\\
&\leq 4\frac{\|H_I\|^2_{\infty,1}}{r} \leq \frac{\epsilon}{2}\;.
\end{align}
Next, let $Q(t_{k+1},t_k)$ denote a channel which implements the three-step procedure outlined in the beginning of the section and let $$\Lambda = \prod_{j=0}^{r-1} Q(t_{j+1},t_j)\;,$$ We claim $\Lambda$ is the desired channel. To show this, note that from \Cref{def:blockencode}, we have upon fixing a signal state $\ket{T} = \ket{0}^m$ and setting $\alpha = 1$ (which can be done since we're implementing qubitization) that 
%%%%%%%%%
\begin{align}
    & \max_{\rho \in D_{2^n}} \left \|\mathcal{E}(t,0)(\rho) - (\bra{0}^m \otimes I_n)(\Lambda(\ket{0}\bra{0}^m \otimes \rho))(\ket{0}^m \otimes I_n) \right \|_{\infty} \leq \max_{\rho \in D_{2^n}}\left \| \mathcal{E}(t,0)(\rho) - \left(\prod_{j=0}^{r-1} \mathcal{U}(t_{j+1},t_j)\right)(\rho) \right \|_{\infty} \nonumber \\
    & + \max_{\rho \in D_{2^n}} \left \|\left(\prod_{j=0}^{r-1} \mathcal{U}(t_{j+1},t_j)\right)(\rho) - (\bra{0}^m \otimes I_n)(\Lambda(\ket{0}\bra{0}^m \otimes \rho))(\ket{0}^m \otimes I_n) \right \|_{\infty} \nonumber \\
    & \leq \frac{\epsilon}{2} + r \max_j \max_{\rho \in D_{2^n}} \left \|\mathcal{U}(t_{j+1},t_j)(\rho) - (\bra{0}^m \otimes I_n)(Q(t_{j+1},t_j)(\ket{0}\bra{0}^m \otimes \rho))(\ket{0}^m \otimes I_n) \right \|_{\infty}\;.
\end{align}

Recall that sampling from $\mathcal{U}(t_k, t_{k+1})$ yields a time-\textit{independent} term $\exp{(-iH_I(\tau_k)/p(\tau_k)})$ where $\tau_k \in [t_k,t_{k+1}] \subset [0,t]$ is a specific time in some sub-interval $[t_k,t_{k+1}]$ at which $H_I(\tau)$ is being evaluated. The latter term above can thus be interpreted as the maximum spectral norm of the difference between an ideal implementation of the time-evolution operator for $t \in [t_k, t_{k+1}]$ and an implementation involving qubitization, maximized over all sub-intervals. This can be made as small as desired via singular value transformation techniques discussed previously. Choosing $$\max_j \max_{\rho \in D_{2^n}} \left \|\mathcal{U}(t_{j+1},t_j)(\rho) - (\bra{0}^m \otimes I_n)(Q(t_{j+1},t_j)(\ket{0}\bra{0}^m \otimes \rho))(\ket{0}^m \otimes I_n) \right \|_{\infty} \leq \frac{\epsilon}{2r} \;,$$ we then have $$\max_{\rho \in D_{2^n}} \left \|\mathcal{E}(t,0)(\rho) - (\bra{0}^m \otimes I_n)(\Lambda(\ket{0}\bra{0}^m \otimes \rho))(\ket{0}^m \otimes I_n) \right \|_{\infty} \leq \frac{\epsilon}{2} + r\frac{\epsilon}{2r} = \epsilon \;.$$ %%%%%

We now define $$\tilde{H}_i(\tau) = H_i(\tau)/p(\tau)\;,$$ 
for $i \neq j$. Using the following identity which holds for all invertible matrices $U$ 
%%%%%%%%%
\begin{equation}
    Ue^AU^{\dag} = \exp(UAU^{\dag}) \;,
\end{equation} 
we have 
%%%%%%%%%
\begin{equation}
\begin{split}
\exp(-iH_I(\tau)/p(\tau)) &= \exp \bigg(e^{iH_j \tau}\bigg(-i\sum_{i\neq j}\tilde{H_i}\bigg)e^{-iH_j \tau}\bigg)\\
&= e^{iH_j \tau}\bigg(\exp\bigg(\sum_{i\neq j}{-i\tilde{H_i}}\bigg)\bigg)e^{-iH_j \tau} \;.\label{eq:expintHam}
\end{split}
\end{equation}

Each $\exp{(-iH_I(\tau_k)/p(\tau_k))}$ term obtained from sampling $\mathcal{U}(t_{k+1},t_k)$ can be expanded as in~\eqref{eq:expintHam}. Using the unitary invariance of the spectral norm, we have the simplification 
\begin{equation*}
\begin{split}
p(\tau_k) &= \frac{\|H_I(\tau_k)\|_{\infty}}{\|H_I(\tau)\|_{\infty,1}} = \frac{\|e^{iH_j \tau_k} (\sum_{i \neq j} H_i) e^{-iH_j \tau_k}\|_{\infty}  }{\int_{t_k}^{t_{k+1}} dt \|e^{iH_j \tau_k} (\sum_{i \neq j}H_i) e^{-iH_j \tau_k}\|_{\infty}} \\
&= \frac{\|\sum_{i \neq j} H_i\|_{\infty}}{\|\sum_{i \neq j}H_i\|_{\infty}\int_{t_k}^{t_{k+1}} dt}  = \frac{1}{t_{k+1}-t_k}    \;.
\end{split}
\end{equation*}

Thus, we obtain a product of terms of the form $$\exp(-iH_I(\tau_k)(t_{k+1}-t_k)) = e^{iH_j \tau_k} \exp{\bigg(-i(t_{k+1}-t_k)\sum_{i \neq j}H_i}\bigg)e^{-iH_j \tau_k}\;.$$ %%%%

%The evolution under $H_j$ can be simulated using a single query to $U_{H_j}$. The middle term can be simulated via qubitization, followed by a singular value transformation to transform the spectrum in~\eqref{eq:walkOP} and recover the original spectrum of $H$. 

We then obtain the overall query complexity by summing~\eqref{eq:qubitcost} as applied to each sub-interval $[t_k, t_{k+1}]$ from $0$ to $r-1$ with error in the QSP transformation at most $\delta$: 
\begin{equation}
    O \bigg( \sum_{k=0}^{r-1} \bigg(\lambda_{\alpha} (t_{k+1}-t_k) + \frac{\log(1/\delta)}{\log \log(1/\delta)}\bigg) \bigg) = O \bigg(\lambda_{\alpha} t + r\frac{\log(1/\delta)}{\log \log(1/\delta)}\bigg)\;.
\end{equation}

Letting $\delta = O(\epsilon/r)$ for our choice of $r$ in the above completes the proof. \end{proof}

Lastly, we consider a hybrid Trotter, qDRIFT, and qubitization I.P. protocol which extends the results of \autoref{cor:hybridTrotter} to include a qubitization step at the end to simulate all the resulting time-independent exponentials. This procedure is largely similar to that outlined in the beginning of the section but with an additional Trotter step:

\begin{enumerate}

\item Move into the interaction frame of a term $H_j$ in $H$ to turn the simulation problem into one involving a time-\textit{dependent} interaction Hamiltonian $H_I(\tau)$ as in~\eqref{eq:intham}.

\item Use the Trotterization technique outlined in \Cref{section:trott} to split the resulting time-ordered exponential into a product of $L$ time-ordered exponentials, one for each summand in the Hamiltonian.  

\item Use continuous qDRIFT to approximate each of the $L$ time-ordered exponentials by the channel~\eqref{eq:qdChannel}. Implementing this channel involves sampling from a probability distribution that yields a product of $r$ time-\textit{independent} terms of the form $\exp{(-iH_I(\tau_k)/p(\tau_k)})$ for each of the $L$ time-ordered exponentials.

\item Use qubitization to simulate the $rL$ time-independent pieces and perform a singular value transformation to transform the spectrum in~\eqref{eq:walkOPmatrix} and recover the original spectrum of $H$.

\end{enumerate}

This yields the following theorem:

\begin{theorem}[Hybrid Trotter, qDRIFT, and Qubitization I.P. Simulation]
\label{thm:trotqDqub}
Let the assumptions of the previous theorem hold. There exists a quantum algorithm such that for any $\epsilon, t > 0$, it implements a quantum channel $\Gamma$ that is a $(1,O(\log L), \epsilon)$ block-encoding  of $e^{-iHt}$ using a number of queries to $\mathrm{PREPARE}$, $\mathrm{SELECT}$, and $W_j(t)$ in
\begin{equation}
    O \bigg(\lambda_{\alpha} t + rL\frac{\log(rL/\epsilon)}{\log \log(rL/\epsilon)}\bigg) \;,\label{eq:trotqDqub}
\end{equation}
where $\lambda_{\alpha} = \sum_{l \neq j} |w_l|$ and $r$ is as in~\eqref{eq:trotqDtstepIP}.

\end{theorem}

\begin{proof}
    Let $\Gamma$ denote a channel which implements the four-step procedure outlined above. Using the notation from the proof of the preceding theorem, we can express $\Gamma$ as $$\Gamma = \prod^r_{j=1} \prod^L_{k=1} Q_k(t_{j+1},t_j)$$ where the subscript $k$ denotes the quantum channel performing steps 3-4 above for a specific Hamiltonian term $H_k$. 
    
    Replicating the arguments of the preceding theorem, we can pick \begin{equation}
        \max_k \max_j \max_{\rho \in D_{2^n}} \left \|\mathcal{U}_k(t_{j+1},t_j)(\rho) - (\bra{0}^m \otimes I_n)(Q_k(t_{j+1},t_j)(\ket{0}\bra{0}^m \otimes \rho))(\ket{0}^m \otimes I_n) \right \|_{\infty} \leq \frac{\epsilon}{2rL} \;.
    \end{equation} 
    
    From \Cref{cor:hybridTrotter}, we can pick $r \geq \frac{2t^2}{\epsilon} \bigg(c_I + 4 \sum_{k \neq l}^L \|H_k\|^2_{\infty} \bigg)$. Then from the triangle inequality, we have 
    %%%%%%
    \begin{align}
    & \max_{\rho \in D_{2^n}} \left \|\mathcal{E}(t,0)(\rho) - (\bra{0}^m \otimes I_n)(\Gamma(\ket{0}\bra{0}^m \otimes \rho))(\ket{0}^m \otimes I_n) \right \|_{\infty} \nonumber \\
    & \leq \max_{\rho \in D_{2^n}}\left \| \mathcal{E}(t,0)(\rho) - \left(\prod_{j=1}^r \prod_{k=1}^{L} \mathcal{U}(t_{j+1},t_j)\right)(\rho) \right \|_{\infty} \nonumber \\
    & + \max_{\rho \in D_{2^n}} \left \|\left(\prod_{j=1}^r \prod_{k=1}^{L} \mathcal{U}(t_{j+1},t_j)\right)(\rho) - (\bra{0}^m \otimes I_n)(\Gamma(\ket{0}\bra{0}^m \otimes \rho))(\ket{0}^m \otimes I_n) \right \|_{\infty} \nonumber \\
    & \leq \frac{\epsilon}{2} + rL \max_k \max_j \max_{\rho \in D_{2^n}} \left \|\mathcal{U}(t_{j+1},t_j)(\rho) - (\bra{0}^m \otimes I_n)(Q(t_{j+1},t_j)(\ket{0}\bra{0}^m \otimes \rho))(\ket{0}^m \otimes I_n) \right \|_{\infty} \nonumber \\
    & \leq \frac{\epsilon}{2} + rL \frac{\epsilon}{2rL} = \epsilon \;.
\end{align}

The overall query complexity is obtained by summing~\eqref{eq:qubitcost} as applied to each of the $\delta = t/r$ sized sub-intervals and summing over the magnitude of the coefficients in the interaction Hamiltonian. We then have, after choosing $\delta = \frac{\epsilon}{2rL}$ that 
%%%%%
\begin{equation}
    O \bigg( r \sum_{k=1}^L \bigg(\lambda_{i} \Delta t + \frac{\log(1/\delta)}{\log \log(1/\delta)}\bigg) \bigg) = O \bigg(\lambda_{\alpha} t + rL\frac{\log(rL/\epsilon)}{\log \log(rL/\epsilon)}\bigg)\;.
\end{equation}
\end{proof}

Note that the above methods can also be used to hybridize these simulation methods in the time-independent case.  Unlike the Trotter-methods, the scaling of the query complexity is not substantially improved. Instead, any potential cost improvements to the simulation come from simplifications to PREPARE and SELECT.

Finally, we note that one can choose other combinations than an outer qDRIFT or Trotter loop and an inner qubitization loop. The first step in each of these hybrid procedures is to exploit the $L^1$-norm invariance of continuous qDRIFT by begining with a time-independent Hamiltonian and transforming into the interaction frame of a particular summand. This results in a time-dependent Hamiltonian, which cannot be simulated via qubitization and constrains us to use either Trotter or qDRIFT first. This still leaves open the possibility of whether trading an inner qDRIFT loop for another Trotterization procedure that decomposes the time-ordered exponentials to ordinary exponentials or randomly interleaving qDRIFT or Trotter procedures can result in additional speedups, and we leave such investigations for future work. 

%%%%%%%%%%%%%%%%%%%%%%%%%%%%%%%%%%%%

\section{Hamiltonian Simulation of Schwinger Model}
\label{section:SchwingerModel}

\subsection{Schwinger Model and Query Complexity Bounds}

We apply these ideas in simulating the Schwinger Model, quantum electrodynamics in 1+1 dimensions on a lattice~\cite{Schwinger1962,COLEMAN1975}. This model has been extensively used as an important stepping stone in simulations of lattice field theories using both tensor networks (see e.g.~\cite{Ba_uls_2013,Pichler2016}) and quantum devices (see e.g.~\cite{Hauke2013,Martinez_2016,Klco2018}).

Using the Hamiltonian formulation of lattice gauge theory in the U(1) compact case~\cite{Kogut1975,Banks1976}, the Hamiltonian of the model with $N-1$ links and $N/2$ spatial sites (half of which are electronic and half are positronic), is given by
\begin{equation}
H  = H_E + H_h + H_M
\end{equation}
%%%%%%%
with 
\begin{align}
    H_E &= \frac{g^2 a}{2} \sum_{r} E_r^2 \label{eq:SchMElectric} \\
    H_h &= \frac{1}{2a} \sum_{r} U_r \psi^{\dag}_r \psi_{r+1} - U_r^{\dag} \psi_r \psi^{\dag}_{r+1} \label{eq:SchInt} \\
    H_M &= m \sum_{r} (-1)^r \psi^{\dag}_r \psi_r \label{eq:SchMag},
\end{align}
%%%%%%%
where $a$ is the lattice spacing, $m$ the fermion mass, and $g$ is the coupling constant. $H_E$ can be interpreted as the electric energy given in terms of $E_r$, the integer-valued electric fields residing on the links. The remaining terms are expressed in terms of the fermionic operators $\psi_r$ and $\psi^{\dag}_r$ living on each site $r$, and the unitary link operators $U_r = e^{i a A_r}$ expressed in terms of the gauge-potential $A_{\mu} = (0,A_1)$ in the temporal-gauge. $H_h$ is a lattice analog of the minimal coupling of the Dirac fermionic field to the gauge field and $H_M$ is the mass energy of the Dirac fermions, which are staggered based on the $(-1)^r$ factor. 

We also have the following commutation relations between the link operators $E_r$ and $U_r$
%%%%%
\begin{align}
    &[E_r, U_s] = U_r \delta_{rs} \Rightarrow [E_r, U^{\dag}_s] = -U_r^{\dag} \delta_{rs}\label{eq:commute},
\end{align}
and between the fermionic creation and annihilation operators
\begin{align}
    \{\psi_r, \psi_s\} = \{\psi_r^{\dag}, \psi_s^{\dag}\} &= 0 \\
    \{\psi_r, \psi_s^{\dag}\} &= \delta_{rs}.
\end{align}

We can map the fermionic creation and annihilation operators in equations~\eqref{eq:SchInt} and~\eqref{eq:SchMag} onto a corresponding set of operators acting on spin degrees of freedom via the Jordan-Wigner transformation
\begin{equation}
\psi_r^{\dag} = \frac{(X_r - i Y_r)}{2} \prod_{j=1}^{r-1} Z_j.
\end{equation}

Substituting the above into~\eqref{eq:SchInt} and~\eqref{eq:SchMag} and simplifying yields $$H_h = \frac{1}{2a} \sum_{r=1}^{N-1} [U_r \sigma^{-}_{r} \sigma^{+}_{r+1} + U_r^{\dag} \sigma^{+}_{r}\sigma^{-}_{r+1}]$$
\begin{equation}
= \frac{1}{8a}  \sum_{r=1}^{N-1} [(U_r + U_r^{\dag})(X_r X_{r+1} + Y_r Y_{r+1}) + i(U_r - U_r^{\dag})(X_r X_{r+1} - Y_r Y_{r+1})] \label{eq:JWSchInt}
\end{equation}

and
\begin{equation}
H_M = \frac{m}{2} \sum_{r=1}^{N} (-1)^{r+1} Z_r. \label{eq:JWSchMag}
\end{equation}
Note that a factor of $I/2$ was dropped in the above equation since terms proportional to the identity in a Hamiltonian merely shift the spectrum by a constant. The derivation above also assumes open boundary conditions, but generalizations to periodic boundary conditions are straightforward. In that case the total number of links becomes $N$ instead of $N-1$ and the asymptotic results we derive below for simulating the Schwinger Model remain unchanged.

It is customary to use the electric eigenbasis $\ket{\epsilon}_r$ for the infinite-dimensional Hilbert space of each link. In this basis, the $E_r$ operator takes the diagonal form $$E_r = \sum_{\epsilon} \epsilon \ket{\epsilon}_r \bra{\epsilon}_r$$ and $U_r$ takes the form $$U_r = \sum_{\epsilon} \ket{\epsilon +1} \bra{\epsilon},$$ i.e. of a raising operator. Note that in order to map these degrees onto a quantum computer, it is customary to truncate the link Hilbert space by wrapping the electric field at a chosen cutoff $\Lambda$. This requires modifying the commutation relations in~\eqref{eq:commute} but this issue is not directly relevant for our present work. 

Since $H_h$ and $H_M$ are manifestly a sum of unitary operators, we can use the PREPARE and SELECT oracles from the qubitization simulation technique outlined previously. The overarching strategy is to move into the interaction frame of the $H_E$ term, employ our hybrid protocols as outlined in the previous sections, and determine the query complexity in terms of the qubitization query model. The physical reasons for selecting the $H_E$ term for the interaction picture is that the spectral norm of $E_r$ is either large for a large cutoff $\Lambda$ or unbounded in the strong coupling regime where $g\to\infty$. Choosing this term ``removes" it from consideration in the interaction Hamiltonian as per equation~\eqref{eq:intham}. Additionally, since the $E_r$ operators are diagonal in its eigenbasis and the matrix elements are computable in polynomial time, the cost of simulating $H_E$ in isolation is in $O({\rm poly}(n\log(1/\epsilon))$~\cite{berry2007efficient}. This efficiency justifies the choice to consider such simulations as oracles in the prior discussion.  As $H_M$ commutes with the $H_E$ term and is also 1-sparse, we may also opt to move into the combined interaction frame of the $H_E$ and $H_M$ terms. In this case, it will suffice to simulate only the $H_h$ term via qubitization, and the simulation of this term will be the biggest asymptotic driver of the query complexity.

Recall that the PREPARE oracle acts on an empty ancilla register of $O(\log L)$ qubits, if $L$ is the number of terms in the decomposition of the Hamiltonian into unitary operators, and prepares the superposition state $$\text{PREPARE} \equiv \sum_{l=1}^L \sqrt{\frac{w_l}{\lambda}} \ket{l}\bra{0},$$
where $w_l$ denotes the coefficients of the terms in the decomposition of $H_h$ and $\lambda = \sum_l |w_l|$ is the sum of the absolute value of the coefficients in the $H_h$ term. Note that this oracle does not get altered when moving into the interaction frame since the coefficients $w_l$ remain the same. Since there are $8(N-1)$ terms in the LCU decomposition of $H_h$, we may set $L = 8(N-1)$. There is only one type of coefficient in $H_h$ in terms of magnitude, so $w_l/\lambda = 1/L$ and we obtain for our situation
%%%%
\begin{equation}
    \text{PREPARE} \equiv \frac{1}{\sqrt{L}} \sum_{l=1}^L \ket{l}\bra{0}. \label{eq:PREPAREInt}
\end{equation}

As a result, we can scale every term in our Hamiltonian by a factor of $8a$ and scale the simulation time by a factor of $1/(8a)$. 

On the other hand, a modification of the traditional select oracle is used to incorporate the interaction picture:
\begin{align}
\text{SELECT}' \equiv  &\sum_l \ket{l}\bra{l} \otimes e^{i(H_E + H_M) t}H'_l e^{-i(H_E + H_M) t} \nonumber \\
&= (I \otimes e^{i (H_E + H_M) t})(\sum_l \ket{l}\bra{l} \otimes H'_l)(I \otimes e^{-i(H_E + H_M) t}). \label{eq:SELECTInt}
\end{align}

This is merely the customary SELECT oracle but conjugated by the unitary operator $e^{i H_E t}e^{i H_M t}$ on the data qubits since $H_E$ and $H_M$ commute. It thus suffices to give circuit implementations of the usual SELECT oracle. 

To summarize, since the conditions of \autoref{thm:IPqDqubitsim} are satisfied, we have the following corollary:

\begin{corollary}[Hybrid qDRIFT and Qubitization I.P Simulation for the Schwinger Model]\label{corollary_SM1}
    Let $H = H_E + H_M + H_h$ be the Schwinger model Hamiltonian as given in~\eqref{eq:SchMElectric},~\eqref{eq:SchInt}, and~\eqref{eq:SchMag}. Then we can perform the Hamiltonian simulation of $H$ with the method of \Cref{thm:IPqDqubitsim} using a total number of queries to PREPARE, SELECT, and $W_{H_M + H_E}(t)$ in 
    %%%%%%%%%%%%%
    \begin{equation}
        O \bigg(\frac{N^2 t^2}{a^2 \epsilon} \frac{\log(Nt/a\epsilon)}{\log\log(Nt/a\epsilon)} \bigg). \label{eq:querO}
    \end{equation}
    %%%%%%%%%%%%
    Here, $N$ is the number of sites in the system, $a$ is the lattice spacing, $t\ge 0$ is the simulation time, and $\epsilon$ is the error quantifying the distance in 1-norm from the ideal time-evolution channel.  %\ale{[Se comment in the proof of \autoref{thm:IPqDqubitsim} about adding the cost of fast-forwarding the "free" terms, at least at the level of oracle calls]}
\end{corollary}

\begin{proof}
    Note that $H_h$ is manifestly a linear combination of unitary operators and that $H_M$ and $H_E$ are diagonal, and therefore 1-sparse, and commute with each other. Thus, the conditions of \autoref{thm:IPqDqubitsim} are satisfied and we may move into the interaction frame of both the $H_M$ and $H_E$ terms. 
    
    To compute the explicit form of $\lambda'$, note that $H_h$ consists of a sum over $8(N-1)$ unitary terms each with a coefficient of absolute value $1/(8a)$. We therefore have $(N-1)/a$ for the overall sum. Thus 
    %%%%%%
    \begin{equation}
    \lambda' = \frac{N-1}{a}.
    \end{equation} 

    Now note that $\|U\|_{\infty} = 1$ for any unitary operator $U$. Then we have $\|H_h\|_{\infty} \leq \lambda'$ by using the triangle-inequality and the sub-multiplicativity of the Schatten infinity norm. Substituting these relationships into~\eqref{eq:mainthm} and retaining the dominant terms gives the claimed query complexity.\end{proof}
    
\begin{corollary}[Hybrid Trotter, qDRIFT, and qubitization I.P Simulation for Schwinger Model]\label{corollary_SM2} Let $H = H_E + H_M + H_h$ be the Schwinger model Hamiltonian as given in~\eqref{eq:SchMElectric},~\eqref{eq:SchInt}, and~\eqref{eq:SchMag}. Then we can perform the Hamiltonian simulation of the Schwinger model with the method of \Cref{thm:trotqDqub} using a number of queries to PREPARE, SELECT, and $W_{H_E + H_M}(t)$ in
    
\begin{equation}
       O \bigg(\frac{Nt^2}{a^2 \epsilon} \frac{\log(Nt^2/(a^2 \epsilon^2))}{\log \log(Nt^2/(a^2 \epsilon^2))} \bigg). \label{eq:querOTrot}
\end{equation}
\end{corollary}

\begin{proof}
    We directly apply \Cref{cor:hybridTrotter} to the situation where we move into the interaction frame of $H_E$ and $H_M$, leaving only the $H_h$ term of the Schwinger model remaining. For the terms in $H_h$, we use the notation $U_r \sigma_r^i \sigma_{r+1}^j$, where $i = 1$ or $2$ so that $\sigma^1_r = X_r$ and $\sigma^2_r = Y_r$. Since $[U_r, U_s] = 0$ for all $r,s$ and $[U_r, U^{\dag}_r] = 0$ since $U_r$ is unitary (and therefore normal), we need only focus on the commutators between the Pauli matrices in computing $c_I$. But since Pauli operators acting on different sites commute, we can further specialize to considering those terms that yield $[X_r, Y_r] = iZ_r$. Therefore, given a term $U_r \sigma_r^i \sigma_{r+1}^j$, it fails to commute with only $U_{r-1} \sigma^k_{r-1} \sigma^l_{r}$ and $U_{r+1} \sigma_{r+1}^m \sigma_{r+2}^n$, with analogous statements holding for $U^{\dag}_r \sigma_r^i \sigma_{r+1}^j$. In other words, the terms involving a particular site $r$ fail to commute with only those involving adjacent sites. Terms such as $[U_r X_r X_{r+1}, U_r Y_r Y_{r+1}]$ do not contribute since $$[U_r X_r X_{r+1}, U_r Y_r Y_{r+1}] = U_r^2 X_r Y_r X_{r+1}Y_{r+1} - U_r^2 Y_r X_r Y_{r+1} X_{r+1} = 0,$$ where we've used in the anti-commutation relation twice $\{X_k, Y_k\} = 0$ to obtain the last equality. 
    
    Therefore, we can decompose $H_h$ into ``even" and ``odd" pieces as follows:
    %%%%%%
    \begin{align*} 
        H_h^{\text{even}} &= \frac{1}{8a}  \sum_{r=1}^{(N-1)/2} [(U_{2r} + U_{2r}^{\dag})(X_{2r} X_{2r+1} + Y_{2r} Y_{2r+1}) + i(U_{2r} - U_{2r}^{\dag})(X_{2r} X_{2r+1} - Y_{2r} Y_{2r+1})] \\
        H_h^{\text{odd}} &= \frac{1}{8a}  \sum_{r=1}^{(N-1)/2} [(U_{2r-1} + U_{2r-1}^{\dag})(X_{2r-1} X_{2r} + Y_{2r-1} Y_{2r}) + i(U_{2r-1} - U_{2r-1}^{\dag})(X_{2r-1} X_{2r} - Y_{2r-1} Y_{2r})]. 
    \end{align*} 
    
    From the preceding discussion, given a particular term in the ``odd" sum, there are exactly two terms in the even sum that fail to commute with it. In particular, there is exactly one term in $H^{\text{even}}_I$ with higher site index that fails to commute with it. Then by the definition of $c_I$, $$c_I = \|[H^{\text{even}}_p, H^{\text{odd}}_q]\| = \frac{(N-1)}{2} \frac{1}{64 a^2} = \frac{N-1}{128 a^2}.$$ 
    
    Similarly, since each term in $H_h$ has norm 1, we have $$4 \sum_{r=1}^{N-1} \|H_h\|^2_{\infty} \le \frac{32(N-1)}{64a^2.}$$ 
    
    Substituting these into~\eqref{eq:trotqDtstepIP} gives
    %%%%%%%%%%
    \begin{equation}
       r \geq \frac{65(N-1)t^2}{128 a^2 \epsilon}.
    \end{equation}

Substituting this and $L=2$ into~\eqref{eq:trotqDqub} and retaining the dominant terms gives
    %%%%%%%%%%5
    \begin{equation}
        O \bigg(\frac{Nt^2}{a^2 \epsilon} \frac{\log(Nt^2/(a^2 \epsilon^2))}{\log \log(Nt^2/(a^2 \epsilon^2))} \bigg)
    \end{equation}
    as claimed. Note that these choices for the errors ensure that the total error for approximation of the ideal time-evolution channel via this entire procedure is $\epsilon/2 + (2r)(\epsilon/4r) = \epsilon/2 + \epsilon/2 = \epsilon$. \end{proof}

Comparing the results of the preceeding corollaries, we see that Trotterizing first before applying qDRIFT can result in improvements in the query complexity in situations where the Hamiltonian has additional commutator structure that can be exploited. For unstructured problems, the additional Trotterization step is not generally useful. 

\subsection{Construction of Prepare and Select Oracles}
\label{ssec:qubitiz_schwinger}

We now give high-level circuit implementations of the aforementioned SELECT and PREPARE oracles for the $H_h$ term of the Schwinger model.

We opt to employ a unary encoding of the control qubits $\ket{c_i}$ needed for the prepare and select oracles, i.e. $\ket{i} = \ket{0...1...0}$, where the $1$ occurs on the $i$-th spot in the ket. Though this encoding requires a number of qubits linear in the number of terms in the LCU decomposition rather than logarithmic for the implementation of the oracles, it greatly simplifies the control structures required within the circuits. 

To implement the select oracle for $H_h$, we exploit certain patterns within the coefficients of the terms in $H_h$. Note that there are terms with $U_r^{\dag}$ and $U_r$, phases of $\pm i$, and $X_r$ and $Y_r$. We can switch between $X$ and $Y$ via the identity $SXS^{\dag} = Y$. From techniques involving two's complement numbers, there exists a unitary operator $Q$ that can flip $U_r$ to $U^{\dag}_r$ \cite{Sanders_2020}. Lastly, the factors $\pm i$ can be inserted via suitable insertions of controlled-Z gate and controlled-S gate operations. The circuit that accomplishes this is given in \autoref{fig:SELECTcirc}. 
%%%%%%%
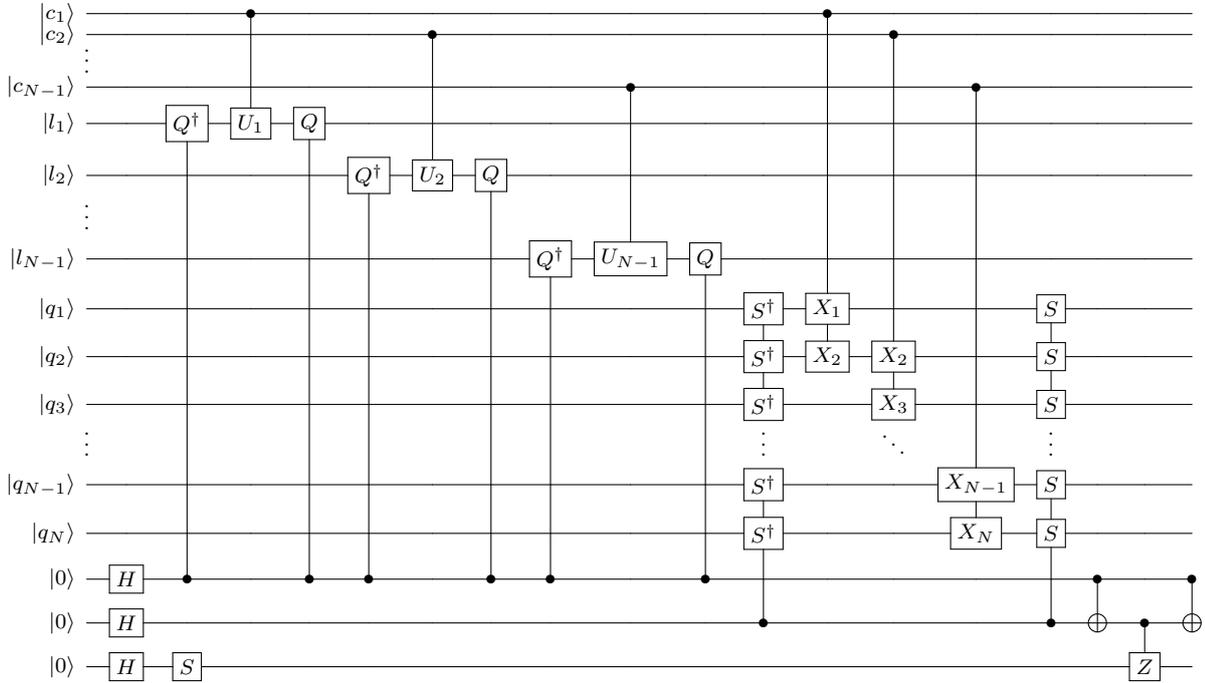
\begin{figure}[ht]
\footnotesize 
    \[
           \Qcircuit @C=1em @R=.7em  {
            \lstick{\ket{c_1}} & \qw & \qw &\ctrl{5} & \qw & \qw & \qw & \qw & \qw & \qw & \qw & \qw & \ctrl{10} & \qw & \qw & \qw & \qw & \qw & \qw \\
            \lstick{\ket{c_2}} & \qw & \qw & \qw & \qw & \qw & \ctrl{5} & \qw & \qw & \qw & \qw & \qw & \qw & \ctrl{10} & \qw & \qw & \qw & \qw & \qw \\
            \vdots & & & & & & & & & & & & & & \\
            & & \\
            \lstick{\ket{c_{N-1}}} & \qw & \qw & \qw & \qw & \qw & \qw & \qw & \qw & \ctrl{5} & \qw & \qw & \qw & \qw & \ctrl{11} & \qw & \qw & \qw & \qw \\
            \lstick{\ket{l_1}} & \qw & \gate{Q^{\dag}} & \gate{U_1} & \gate{Q} & \qw & \qw & \qw & \qw & \qw & \qw & \qw & \qw & \qw & \qw & \qw & \qw & \qw & \qw \\
            \lstick{\ket{l_2}} & \qw & \qw & \qw & \qw & \gate{Q^{\dag}} & \gate{U_2} & \gate{Q} & \qw & \qw & \qw & \qw & \qw & \qw & \qw & \qw & \qw & \qw & \qw \\
            \vdots & & \\
            & & \\
            \lstick{\ket{l_{N-1}}} & \qw & \qw & \qw & \qw & \qw & \qw & \qw & \gate{Q^{\dag}} & \gate{U_{N-1}} & \gate{Q} & \qw & \qw & \qw & \qw & \qw & \qw & \qw & \qw \\ 
            \lstick{\ket{q_1}} & \qw & \qw & \qw & \qw & \qw & \qw & \qw & \qw & \qw & \qw & \gate{S^{\dag}} & \gate{X_1} & \qw & \qw & \gate{S} & \qw & \qw & \qw \\
            \lstick{\ket{q_2}} & \qw & \qw & \qw & \qw & \qw & \qw & \qw & \qw & \qw & \qw & \gate{S^{\dag}} \qwx & \gate{X_2} \qwx & \gate{X_2} & \qw & \gate{S} \qwx & \qw & \qw & \qw \\
            \lstick{\ket{q_3}} & \qw & \qw & \qw & \qw & \qw & \qw & \qw & \qw & \qw & \qw & \gate{S^{\dag}} \qwx & \qw & \gate{X_3} \qwx & \qw & \gate{S} \qwx & \qw & \qw & \qw \\
            \vdots & & & & & & & & & & & \vdots & & \ddots & & \vdots \\
            & & & & & & & & & & & \\
            \lstick{\ket{q_{N-1}}} & \qw & \qw & \qw & \qw & \qw & \qw & \qw & \qw & \qw & \qw & \gate{S^{\dag}} & \qw & \qw & \gate{X_{N-1}} & \gate{S} & \qw & \qw & \qw \\
            \lstick{\ket{q_N}} & \qw & \qw & \qw & \qw & \qw & \qw & \qw & \qw & \qw & \qw & \gate{S^{\dag}} \qwx & \qw & \qw & \gate{X_N} \qwx & \gate{S} \qwx & \qw & \qw & \qw \\
            \lstick{\ket{0}} & \gate{H} & \ctrl{-12} & \qw & \ctrl{-12} & \ctrl{-11} & \qw & \ctrl{-11} & \ctrl{-8} & \qw & \ctrl{-8} & \qw & \qw & \qw & \qw & \qw & \ctrl{1} & \qw & \ctrl{1} \\
            \lstick{\ket{0}} & \gate{H} & \qw & \qw & \qw & \qw & \qw & \qw & \qw & \qw & \qw & \ctrl{-2} & \qw & \qw & \qw & \ctrl{-2} & \targ & \ctrl{1} & \targ \\
            \lstick{\ket{0}} & \gate{H} & \gate{S} & \qw & \qw & \qw & \qw & \qw & \qw & \qw & \qw & \qw & \qw & \qw & \qw & \qw & \qw & \gate{Z} & \qw 
            }
    \]
\caption{SELECT circuit for the $H_h$ term in Schwinger Model} 
\label{fig:SELECTcirc} 
\end{figure}

Here, $\ket{c_r}$ represent the control qubits used to implement the controlled operations, $\ket{l_r}$ the qubits corresponding to the links of the system, and $\ket{q_r}$ the qubits corresponding to the sites. Note how the control gate structure in this unary encoding is much more simple than what would have been required with a binary encoding. Though the latter encoding would have required only $\log(N-1)$ control qubits instead of $N-1$ as above, the advantage there is mitigated by the fact that numerous multi-controlled gates would have been required. 

Since the terms in $H_h$ for a given $r$ only differ by coefficients of $\pm 1$ or $\pm i$, these can be implemented via the insertion of a $Z$ or $S$ gate through the above constructions that exploit the aforementioned patterns in $H_h$. Note that the number of qubits needed to specify the state of link $\ket{l_r}$ will depend on the cut off for the electric energy term chosen. If our cutoff is $\Lambda$, then $\ket{l_r}$ will be a $\log \Lambda$-qubit state and $U_r$ a $\log \Lambda$-qubit operator. 

The circuit implementation of the PREPARE oracle reduces to preparing a uniform superposition state in binary, as per~\eqref{eq:PREPAREInt}, and then converting the encoding to a unary one. The overall circuit with the general pattern is depicted in \autoref{fig:PREPAREcirc} with $k = \log(8(N-1))$ ancilla control qubits. Note that since in the unary encoding an integer $k$ is expressed as a state with a $1$ in the $k$-th spot and 0's elsewhere, the $X$ gate on $\ket{0}_1$ and the subsequent swaps have the effect of permuting the $1$ to the appropriate position. 

\begin{figure}[ht]
    \[
           \Qcircuit @C=1em @R=.7em {
            \lstick{\ket{b_1}} & \gate{H} & \ctrlo{3} & \ctrlo{3} & \ctrlo{3} & \qw & \qw & \qw \\
            & \vdots & & & & \ddots & & \\
            & & \\
            \lstick{\ket{b_{k-1}}} & \gate{H} & \ctrlo{1} & \ctrl{1} & \ctrl{1} & \qw & \qw & \qw \\
            \lstick{\ket{b_k}} & \gate{H} & \ctrl{1} & \ctrlo{2} & \ctrl{2} & \qw & \qw & \qw \\
            \lstick{\ket{c}_1} & \gate{X} & \qswap & \qswap & \qw & \qw & \qw & \qw \\
            \lstick{\ket{c}_2} & \qw & \qswap \qwx & \qw \qwx & \qswap & \qw & \qw & \qw \\
            \lstick{\ket{c}_3} & \qw & \qw & \qswap \qwx & \qw \qwx & \qw & \qw & \qw \\
            \lstick{\ket{c}_4} & \qw & \qw & \qw & \qswap \qwx & \qw & \qw & \qw \\
            & \vdots & & \vdots & & \cdots \\
            & \\
            \lstick{\ket{c}_{N-1}} & \qw & \qw & \qw & \qw & \qw & \qw & \qw
            }
    \]
\caption{PREPARE circuit for the $H_h$ term in Schwinger Model} 
\label{fig:PREPAREcirc} 
\end{figure}
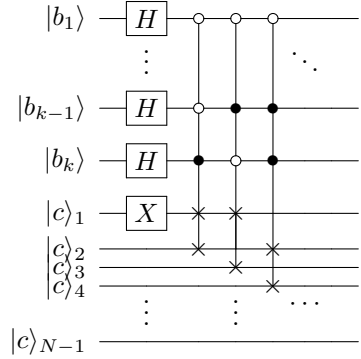

The control structure on the bits $b_j$ encoding the binary integers for the controlled-swap gates is given precisely by the binary representation of that index. For example, since the binary integer $\ket{00\ldots1}$ gets mapped to $\ket{01\ldots 0}$ in our unary encoding, we have to do a swap on the $\ket{0}_1$ and $\ket{0}_2$ qubits controlled on the first $k-1$ $b$ qubits being 0 and $b_k$ being 1. 

\subsection{Gate Complexity Analysis}

We now analyze the gate complexity per query to $\text{CTRL}(\mathcal{W})$ of our simulation protocol and do so by analyzing the circuits given in \autoref{fig:SELECTcirc} and \autoref{fig:PREPAREcirc}. It suffices to express the complexity in terms of Toffoli gates since they dominate the computational complexity compared to Clifford operations. 

Our analysis proceeds as follows: 

\begin{enumerate}
    \item First consider the PREPARE$^{\dag}$ and PREPARE operations in \autoref{fig:walkCirc}. We make the approximation that they have roughly the same gate complexity and that it therefore suffices to determine the gate complexity of just the PREPARE circuit. 
    
    From \autoref{fig:PREPAREcirc}, note that we have $N-1$ multiply-controlled swap gates since we needed to perform a binary-to-unary conversion to the $N-1$ qubits we have. Each can be converted to standard $C^k(\text{SWAP})$ by inserting $X$ gates on either side of the 0 controls. Since we are assuming Pauli operations are approximately cost-free, it suffices to determine the gate complexity of these $N-1$ $C^k(\text{SWAP})$ gates. Standard circuit arguments show that the following identities hold:
    $$
    
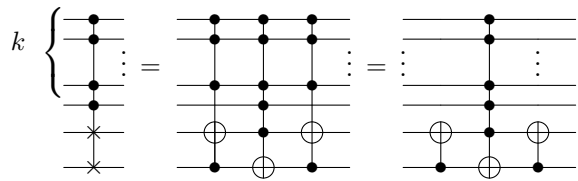
\Qcircuit @C=1em @R=.5em {
    & \ctrl{1} & \qw &&& \ctrl{1} & \ctrl{1} & \ctrl{1} & \qw &&& \qw & \ctrl{1} & \qw & \qw \\
    & \ctrl{3} & \qw &&& \ctrl{3} & \ctrl{3} & \ctrl{3} & \qw &&& \qw & \ctrl{3} & \qw & \qw \\
    & & \vdots & & & & & & \vdots & & \vdots & & & \vdots \\
    & & & = & & & & & & = & & & & \\
    & \ctrl{1} & \qw &&& \ctrl{3} & \ctrl{1} & \ctrl{2} & \qw &&& \qw & \ctrl{1} & \qw & \qw \\
    & \ctrl{1} & \qw &&& \qw & \ctrl{1} & \qw & \qw &&& \qw & \ctrl{1} & \qw & \qw
    \inputgroupv{1}{5}{.8em}{.8em}{k} \\
    & \qswap & \qw &&& \targ & \ctrl{1} & \targ & \qw &&& \targ & \ctrl{1} & \targ & \qw \\
    & \qswap \qwx & \qw &&& \ctrl{-1} & \targ & \ctrl{-1} & \qw &&& \ctrl{-1} & \targ & \ctrl{-1} & \qw
    }
    $$

    Thus, each $C^k(\text{SWAP})$ gate can be decomposed into a $C^{k+1}(\text{NOT})$ gate and 2 CNOT gates. From Corollary 1 in \cite{He_Toffoli}, we get that $C^{k+1}(\text{NOT})$ gate can be decomposed into $8(k+2)-24 = 8k - 8$ Toffoli gates using only a single auxiliary qubit that can be reused. This gives roughly $2(N-1)(8k-8)$ Toffoli gates that are needed for both the PREPARE and PREPARE$^{\dag}$ parts of \autoref{fig:walkCirc}. 
    Overall we have a gate complexity of $$O(N \log N)$$ with 1 ancilla qubit needed. 
    
    The methods in \cite{bausch2021} can be used to perform PREPARE circuit with $N \log N$ controlled swap gates, each of which can be decomposed into at most four non-Clifford operations using~\cite{jones2013low}. This results in a smaller gate cost but ultimately does not affect the asymptotic gate complexity shown above. 
    
    \item The multiply-controlled $Z$ gate in \autoref{fig:walkCirc} is controlled only on the $N-1$ control qubits that make up the unary encoding portion of \autoref{fig:PREPAREcirc}. Applying the above corollary again, we get a gate complexity of $$O(N)$$ with 1 ancilla qubit needed. 
    
    \item To analyze the gate complexity of the controlled-SELECT operation in \autoref{fig:walkCirc}, we first examine the operations that don't involve $U_r$ and $Q$. The external control distributes among the operations involving the $2(N-1)$ controlled-$S$ and controlled-$S^{\dag}$ gates, $2(N-2) + 2$ controlled-$X_r$ gates, and the final 3 CNOT operations. Adding these together yields a total of $4N -1$ Toffoli gates, giving a gate complexity of $$O(N)$$ with no ancilla qubits needed. 
    
    \item For the $U_r$ gates, note that its action on states is that of an incrementer. This can be implemented by utilizing the quantum ripple-carry adder circuit given by Cuccaro~\cite{cuccaro2004new}. Since we increment $\log \Lambda$-qubit numbers, the gate complexity is given in the paper to be $2\log \Lambda -1$ Toffoli gates, $5\log \Lambda -3$ CNOTs, and $2\log \Lambda - 4$ negations. 
    
    Since we must perform a controlled-$U_r$ operation for the walk operator, we get $(5 \log \Lambda -3)$ Toffoli gates, and $(2 \log \Lambda - 1)$ $C^3$(NOT) gates. Again from Corollary 1 in \cite{He_Toffoli}, the latter is equivalent to $8(2\log \Lambda -1)$ Toffoli gates with $1$ extra qubit needed. Since we must perform $N-1$ of these $U_r$ operations, we have altogether $(N-1)(21 \log \Lambda -11)$ Toffoli gates, giving a gate complexity of $$O(N \log \Lambda)$$ with 1 ancilla qubit needed.

    \item The $Q$ and $Q^{\dag}$ operations can be implemented with $O(\log(\Lambda))$ Toffoli gates each (see \cite{Sanders_2020}), with the extra controls giving only constant pre-factors to the cost. Since there are $2(N-1)$ of these operations performed, we have a gate complexity of $$O(N \log \Lambda)$$
    
    \item Finally, the cost of performing the diagonal Hamiltonian simulation is $O(N \log^2(\Lambda))$ as the computation of the diagonal elements of the Hamiltonian involves squaring the input value, which can be performed in time $O(\log^2(\Lambda))$ (see Lemma 2 of~\cite{Shaw_2020}).
    
\end{enumerate}

These considerations give us the following corollaries:

\begin{corollary}[Gate Complexity for Hybrid qDRIFT and Qubitization I.P. Simulation of Schwinger Model]
    Let $H = H_h + H_M + H_E$ be the Schwinger model Hamiltonian as given in~\eqref{eq:SchMElectric},~\eqref{eq:SchInt}, and~\eqref{eq:SchMag}. Then the Hamiltonian simulation of $H$ can be performed using the method of \Cref{thm:IPqDqubitsim} with a gate complexity in 
    
   \begin{equation}
        O \bigg(\frac{N^3 t^2}{a^2 \epsilon} \frac{\log(Nt/a\epsilon)}{\log\log(Nt/a\epsilon)} \log^2(N\Lambda) \bigg), \label{eq:gatecomplexO}
    \end{equation} 
%%%%%%
    with an ancilla qubit overhead of O(1). In $\tilde{O}$ notation, the gate complexity is  
    %%%%%
    \begin{equation}
        \tilde{O}\bigg(\frac{N^3 t^2}{a^2 \epsilon} \log^2 \Lambda \bigg). \label{eq:gatecomplexOtilde}
    \end{equation}
    
\end{corollary}

\begin{proof}
    Summing up the scaling results from steps 1-6 as outlined above, we get a \textit{per-query} gate complexity of 
    \begin{equation}
    O(N \log^2(\Lambda) + N\log(N)) \subseteq O(N \log^2(N\Lambda) ).
    \end{equation}
    Multiplying these by~\eqref{eq:querO} and retaining the dominant terms yields the stated results. 
\end{proof}

\begin{corollary}[Gate Complexity for Hybrid Trotter, qDRIFT, Qubitization I.P. Simulation of Schwinger Model]\label{cor:schw_best_hybrid}  Let $H = H_h + H_M + H_E$ be the Schwinger model Hamiltonian as given in~\eqref{eq:SchMElectric},~\eqref{eq:SchInt}, and~\eqref{eq:SchMag}. Then the Hamiltonian simulation of $H$ can be performed using the method of \Cref{thm:trotqDqub} with a gate complexity in 
%%%%%%%%
    \begin{equation}
        O \bigg(\frac{N^2t^2}{a^2 \epsilon} \frac{\log(Nt^2/(a^2 \epsilon^2))}{\log \log(Nt^2/(a^2 \epsilon^2))} \log^2(\Lambda) \bigg), \label{eq:gatecomplexTrotO}
    \end{equation} 
%%%%%%
    with an ancilla qubit overhead of O(1). In $\tilde{O}$ notation, the gate complexity is  
    
    \begin{equation}
        \tilde{O}\bigg(\frac{N^2t^2}{a^2 \epsilon} \log^2(\Lambda) \bigg). \label{eq:gatecomplexTrotOtilde}
    \end{equation}

\end{corollary}

\begin{proof}
    Multiplying~\eqref{eq:querOTrot} by $N \log(N\Lambda)$ gives the big-$O$ cost. Dropping all sub-dominant logarithmic factors gives the $\tilde{O}$ scaling. 
\end{proof}

%It is worth noting that the single qubit rotations in \autoref{fig:expreflcirc} have a sub-dominant cost to synthesize since the cost of the altering phase modulation sequence is incorporated within the $\log(1/\epsilon)$ factors in~\eqref{eq:qubitcostgeneral}. As such, the cost of synthesizing such rotations can be at worst  $\log(1/n\epsilon)$ to achieve the desired qubitization error, where $n$ is the number of phase factors in \autoref{fig:phasemodcirc}. 

\subsection{Comparison with Trotterization} 

%\textbf{[Rewrite this section]}

The results of~\eqref{eq:gatecomplexOtilde} and~\eqref{eq:gatecomplexTrotOtilde} can be directly compared with the result given in Corollary 11 of~\cite{Shaw_2020} for using a second-order Trotter-Suzuki formula to perform the quantum simulation of the Schwinger model. We will first consider the regime in which simulations are carried at constant $1/(ga)=O(1)$ and for fixed $m/g=O(1)$. For this condition, we can use the result in Corollary 9 of~\cite{Shaw_2020} which after rescaling the time variable $T \rightarrow ag^2 t/2$ to align with our normalization conventions for the Schwinger Model Hamiltonian, gives a total $T$-gate cost of $$\tilde{O}\bigg(\frac{N^{3/2} t^{3/2} \Lambda a^{1/2} g^2}{\epsilon^{1/2}}\bigg) = \tilde{O}\bigg(\frac{N^{3/2} t^{3/2} \Lambda g^{3/2}}{\epsilon^{1/2}}\bigg)\;.$$

In the same regime, the result of~\eqref{eq:gatecomplexTrotOtilde} gives a gate complexity in
$$\tilde{O}\bigg(\frac{N^2 t^2}{a^2\epsilon}\log^2(\Lambda)\bigg)=\tilde{O}\bigg(\frac{N^2 t^2g^2}{\epsilon}\log^2(\Lambda)\bigg)\;.$$
We then see that the hybrid I.P. scheme provides a quasi-exponential speedup with respect to the electric cutoff $\Lambda$ over the second-order Trotter-Suzuki approach, at the expense of a slightly worse scaling in all the other parameters $(N,t,g,\epsilon)$.

In order to extract physical observables however, it is important to consider that the number of sites $N$ and the the lattice spacing cannot be chosen independently as the product $L=Na$ gives the physical size of the system. In past numerical simulations it was found that choosing $Nga=O(10)$ is appropriate for a large number of configurations (see e.g.~\cite{Ba_uls_2013}). In order to keep our derivation general, we will then introduce the dimensionless parameter $l=Nga$. In addition to the thermodynamic limit $Na\to\infty$, one also has to work with $ga\to0$ in order to recover the continuum limit of the theory. The resulting gate complexity of the hybrid I.P. algorithm from Corollary~\ref{cor:schw_best_hybrid} is found to be $$\tilde{O}\bigg(\frac{l^2 t^2}{g^2a^4\epsilon}\log^2(\Lambda)\bigg)\;,$$
while for the second order Trotter-Suzuki scheme we have to consider two distinct regimes
\begin{itemize}
    \item large cutoff limit $\Lambda g a >1$, in which case the $T$-gate count is bounded by $$\tilde{O}\bigg(\frac{l^{3/2} t^{3/2} \Lambda}{a^{3/2}\epsilon^{1/2}}\bigg)\;,$$%=\tilde{O}\bigg(\frac{l^{3/2} t^{3/2} \Lambda^{5/2}g^{3/2}}{\epsilon^{1/2}}\bigg)\;,$$
    \item small lattice spacing limit $\Lambda g a < 1$, in which case the result of Corollary 9 of~\cite{Shaw_2020} does not hold anymore. The T-gate count can be found instead by using the more general result from Corollary 11 there, resulting in the $\Lambda$-independent scaling $$\tilde{O}\bigg(\frac{l^{3/2} t^{3/2}}{g^{3/2}a^3\epsilon^{1/2}}\bigg)\;.$$
\end{itemize}

These results show that for fixed error $\epsilon$ and lattice extent $l$, the hybrid I.P. approach can be especially beneficial in the first regime thanks to the poly-logarithmic dependence on the cutoff $\Lambda$. In the small lattice spacing regime relevant for the continuum limit, the second order Trotter-Suzuki scheme developed in~\cite{Shaw_2020} has instead a better scaling with respect to all the parameters.

%\ale{[This last part might or might not be interesting. It is kind of rushed and without the same formal rigor used in the derivations above. If you guys want we can just comment about the idea in words and say it would be an interesting exercise to the reader.]}
Finally, for the continuum limit it might be possible to improve the the gate complexity in some regimes by choosing to perform the I.P. simulation in the rotating frame given by the the interaction Hamiltonian $H_h$ instead. Using the hybrid Trotter and qDRIFT scheme from Corollary~\ref{cor:hybridTrotter}, together with the implementation via qubitization of the time evolution operator for $H_h$ derived in Section~\ref{ssec:qubitiz_schwinger}, this scheme has gate cost in$$\tilde{O}\bigg(\frac{N^2 t}{a}+\frac{N^2t^2}{\epsilon}\left(m^2+g^4a^2\Lambda^4\right)\bigg)\;.$$
For $m/g=O(1)$ and introducing the dimensional lattice size $l=Nga$, this becomes$$\tilde{O}\bigg(\frac{l^2 t}{g^2a^3}+\frac{l^2t^2}{a^2\epsilon}+\frac{l^2t^2}{\epsilon}g^2\Lambda^4\bigg)\;.$$
This is clearly worse than either the Hybrid approaches discussed above or the second order Trotter-Suzuki scheme from~\cite{Shaw_2020} in the large cutoff limit $\Lambda ga>1$, but can become competitive in the small lattice spacing limit $\Lambda ga<1$ for some choices of $(\epsilon,l,\Lambda)$. %long enough times $t\gg1/g$ so that $1/g=O(t)$. In this case we have in fact that the second order Trotter-Suzuki formula gives a gate count in
%$$\tilde{O}\bigg(\frac{l^{3/2} t^{3}}{a^3\epsilon^{1/2}}\bigg)\;,$$
%whereas the hybrid algorithm just described offers a potentially better scaling bounded by$$\tilde{O}\bigg(\frac{l^2 t}{g^2a^3}+\frac{l^2t^2\Lambda^2}{a^2\epsilon}\bigg)=\tilde{O}\bigg(\frac{l^2 t^3}{a^3}+\frac{l^2t^2\Lambda^2}{a^2\epsilon}\bigg)\;.$$

A detailed comparison of these different schemes to extract continuum quantities of physical interest in the Schwinger model with some target precision $\delta$ would require a more careful analysis of the scaling of the lattice size $l$ and the electric field cutoff $\Lambda$, as well as a more careful consideration of the logarithmic factors hidden by the $\tilde{O}$ notation. We leave this interesting extension of the present work to future studies.

\section{Collective Neutrino Oscillations}
\label{section:neutrinoosc}

In extreme astrophysical environments, such as supernova explosions, neutrinos are present in such large densities that neutrino-neutrino interactions can become important to describe flavor evolution~\cite{PANTALEONE1992,Duan2006}. These interactions are responsible for the appearance of collective modes in flavor oscillations and have traditionally been studied with the help of a mean-field approximation (see e.g.~\cite{Duan2010,Chakraborty2016b} for reviews). Due to the presence of interactions, many-body effects and quantum correlations could be important in understanding these phenomena and a number of studies is underway with a variety of techniques: from exact diagonalization~\cite{Rrapaj2020} to Bethe-ansatz solutions~\cite{Cervia2019}, from tensor networks~\cite{Roggero2021A,Roggero2021B} to digital quantum simulations~\cite{Hall2021,yeteraydeniz2021collective}. Quantum computing might offer a promising route to study these phenomena in situations where the entanglement entropy grows too fast with system size for tensor network simulations to remain feasible.

An important obstacle towards describing collective oscillations in realistic regimes is the fact that besides interactions with other neutrinos, scattering with external leptons (especially the abundant electrons) is an important effect near the proto-neutron star. The matter interaction terms can become the dominant contributions in this regime, requiring very small time-steps for an accurate simulation of the flavor dynamics. On the quantum computing side, this requirement translates into a large number of gates required for the simulation and it is therefore important to design simulation algorithms with a gentle computational scaling with the external matter density.

The Hamiltonian we are interested in can be written as follows (see e.g.~\cite{Pehlivan2011} for a derivation)
\begin{equation}
H = \sum_{i=1}^{N} \frac{\omega_i}{2} \vec{B}\cdot\vec{\sigma}_i + \frac{\lambda}{2} \sum_{i=1}^N Z_i + \frac{\mu}{2N}\sum_{i<j}^N J_{ij} \vec{\sigma}_i\cdot\vec{\sigma}_j\;.
\end{equation}
Here $\vec{\sigma}_i$ is the vector of Pauli matrices acting on the $i$-th qubit and the single particle energies $\omega_i$ are positive for neutrinos and negative for anti-neutrinos. The coupling matrix $J_{ij}$ takes values in $[0,2]$ and the normalized vector $\vec{B}$ contains the vacuum mixing angle as $\vec{B}=(\sin(2\theta),0,-\cos(2\theta))$. The constants are given by $\lambda = \sqrt{2}G_Fn_e$ and $\mu=\sqrt{2}G_Fn_\nu$, with $G_F$ Fermi's constant and $n_e$ and $n_\nu$ the electron and neutrino densities respectively. In typical situations the electron contribution $\lambda$ is the dominant term. A standard approach to deal with this problem is to move to the rotating frame defined by the unitary $U_e(t)=\exp(-i\frac{t}{2}\lambda\sum_{i=1}^N Z_i)$ and define the Hamiltonian in the interaction picture as
\begin{equation}
\begin{split}
\label{eq:time-dep-ham}
H(t) &= U_e^\dagger(t) H U_e(t) - iU_e^\dagger(t)\frac{\partial}{\partial t} U_e(t)\\ 
&= \sin(2\theta)\sum_{i=1}^N \frac{\omega_i}{2}\left(\cos(\lambda t)X_i-\sin(\lambda t)Y_i\right) -\cos(2\theta)\sum_{i=1}^{N} \frac{\omega_i}{2} Z_i + \frac{\mu}{2N}\sum_{i<j}^N J_{ij} \vec{\sigma}_i\cdot\vec{\sigma}_j\;\\
&= e^{i \lambda \sum_i Z_i t} H_\nu e^{-i \lambda \sum_i Z_i t},
\end{split}
\end{equation}
where
\begin{equation}
H_{\nu} = \sum_{i=1}^{N} \frac{\omega_i}{2} \vec{B}\cdot\vec{\sigma}_i + \frac{\mu}{2N}\sum_{i<j}^N J_{ij} \vec{\sigma}_i\cdot\vec{\sigma}_j\;.
\end{equation}

Typically only the leading order contribution in the Magnus expansion is retained, giving the time-independent Hamiltonian
\begin{equation}
\begin{split}
H_0 &= -\cos(2\theta)\sum_{i=1}^{N} \frac{\omega_i}{2} Z_i + \frac{\mu}{2N}\sum_{i<j}^N J_{ij} \vec{\sigma}_i\cdot\vec{\sigma}_j\;.
\end{split}
\end{equation}
In this limit, flavor states will not experience oscillations and typically this is solved by defining the flavor axis to be rotated by a small phenomenological amount away from the Z axis. It would be desirable however to be able to exercise more control in this approximation. Expansions to high orders in the Magnus expansion quickly produce higher order interactions which will complicate the implementation of the corresponding time-independent evolution. Here we use the time-dependent algorithm described above to work directly in the interaction picture without introducing uncontrollable errors.

%Let's start by assuming we have an implementation of the time evolution operator $U_0(t)$ performing evolution under the background-free Hamiltonian 

\subsection{Trotter Suzuki Approximations in Interaction Frame}
As a first step, let us consider simulating the Hamiltonian in the interaction frame using a $k^{\rm th}$-order Trotter-Suzuki formula such as those in~\cite{wiebe2010higher}.
%As $U_0$ is a diagonal unitary, we can fast forward its evolution and so for simplicity we assume that $e^{-i\lambda t \sum_i Z_i}$ can be simulated using a single query.
%Using a Trotter-Suzuki decomposition of order $K$ to approximate the time-ordered exponential the query complexity in $U_0(t)$ is directly connected with the following quantity
To do this, we need to introduce a notion of the typical energy scale of the time-dependent Hamiltonian with respect to the Trotter decomposition of the interaction frame Hamiltonian. 
If we use a conventional Trotter decomposition, as opposed to~\eqref{eq:c12def}, we find that the error incurred from using a first-order Trotter formula for an ordered operator exponential $U_1(t)$ formed by evaluating the Hamiltonian at $t=0$ and then Trotterizing the resultant ordinary operator exponential is
\begin{equation}
    \left\|\exp_\tau\left(-i\int H(t) \mathrm{d}t\right) - U_1(t)\right\|_\infty \in O((\max_t \|H'(t)\|_\infty + \sum_{p,q} \max_{t}\|[H_p(t),H_q(t)] \|_\infty)t^2),
\end{equation}
where the specific constants can be found using the techniques in~\cite{wecker2015solving}.
The derivative of the Hamiltonian in the interaction frame is in
\begin{equation}
    \|H'(t)\|_\infty \in O(\theta N \lambda).
\end{equation}
The commutator sum similarly obeys
\begin{align}
    \|\sum_{p,q} \max_{t,t'}\|[H_p(t'),H_q(t)] \|_\infty &\in O\left({N\omega\theta + \mu N \theta + \omega N \mu + N\mu^2 } \right)\nonumber\\
    &= O\left(N(\theta(\omega+\mu) +\mu(\omega+\mu) ) \right) \nonumber\\
    &\subseteq O\left(N\mu^2 ) \right),\label{eq:sumHpHq}
\end{align}
where $\omega = \max_i |\omega_i|$ and in the last term we take $\mu \gg \omega$.

The overall error in the simulation is therefore
\begin{equation}
    O\left( ({N\mu^2 + \theta N \lambda})t^2 \right).
\end{equation}
If we break the overall evolution into $r$ time slices, then it follows that the error in the simulation can be made at most $\epsilon$ by choosing
\begin{equation}
    r \in O\left(\frac{N({\mu^2 + \theta  \lambda})t^2}{\epsilon} \right).
\end{equation}
As there are $O(N^2)$ operator exponentials per time step, the total number of operator exponentials needed to perform the simulation is
\begin{equation}
    N_{\exp} \in O\left(\frac{N^3({\mu^2 + \theta  \lambda})t^2}{\epsilon}  \right).
\end{equation}
Since each operator exponential requires $O(1)$ gates from the $H, R_z, {\rm CNOT}$ gate library, the gate complexity is also proportional to this~\cite{berry2007efficient}. Interestingly, using the swap-network protocol from Ref.~\cite{Hall2021} (and inspired from their fermionic variant~\cite{Kivlichan2018}), this cost is not affected by limited connectivity in the device despite the interaction being all-to-all. This cost also coincides with the optimal scaling with $\lambda$ permitted by the no-fast forwarding theorem~\cite{Low_2016}, despite being a low-order formula that has inferior scaling with respect to the other parameters relative to alternative simulation methods.

Higher-order \textit{time-dependent} Trotter formulas for the simulation can be used, but the advantage gleaned by using them with respect to the $\lambda$ scaling is less clear. Such algorithms scale with the parameter~\cite{wiebe2010higher}
\begin{equation}
    \Lambda_k^{k+1}/r^{k} = \max_{j\le k}(\|\partial_t^j H(t)\|_{\infty}^{1/j+1})^{k+1}/r^k \in O(\lambda^{k}/r^k),
\end{equation} where $k$ represents the order of the Trotter formula. It then follows that these formulas ultimately lead to the same linear scaling in the gate complexity with $\lambda$ (assuming that $\theta\in \Theta(1)$).  By contrast, if we did not use the interaction picture algorithm, the cost of simulation using the $k^{\rm th}$-order Trotter \textit{time-independent} formula would scale as $O(\lambda^{1+1/2k})$~\cite{berry2007efficient,wiebe2010higher}. This illustrates that for problems with an imbalance in the scales of the operators, switching to an interaction frame can be beneficial at virtually no cost overhead.

%with $H(t)$ the time-dependent Hamiltonian from Eq.~\eqref{eq:time-dep-ham}. Since the time-dependence comes from the use of interaction-picture representation, the derivatives can be expressed in therms of commutators between $H_\nu$ and $H_e=\frac{\lambda}{2}\sum_{i=1}^N Z_i$. For instance with $k=1$ we have
%\begin{equation}
%\begin{split}
%\|H(t)\| &= \|H_\nu\| \leq \frac{1}{2}\left(\sum_{i=1}^N |\omega_i| + \mu (N-1)\right)\\
%\|H^{(1)}(t)\| &\le \|\left[H_e,H_\nu\right]\| \leq \frac{\lambda\sin(2\theta)}{2} \sum_{i=1}^N |\omega_i| +\frac{\lambda\mu}{4N} \sum_{i\neq j} J_{ij}2 \|\left[Y_iX_j-X_iY_j+X_jY_i-Y_jX_i\right]\|\\
%&\leq \frac{\lambda\sin(2\theta)}{2} \sum_{i=1}^N |\omega_i| + \frac{2\lambda\mu}{N} \sum_{i\neq j} J_{ij}\leq  \frac{\lambda\sin(2\theta)}{2} \sum_{i=1}^N |\omega_i| + 4 \lambda\mu (N-1)
%\end{split}
%\end{equation}
%Since $\lambda\gg|\omega_i|$ the derivative term dominates and we have
%\begin{equation}
%\Lambda_1 \leq \sqrt{\frac{\lambda\sin(2\theta)}{2} \|\omega_i\|_1 + 4 \lambda\mu (N-1)} = \mathcal{O}(\sqrt{\lambda\mu N})\;.
%\end{equation}

%In this case the oracle cost of the time dependent simulation scales as $\mathcal{O}\left(\frac{\lambda\mu N t^2}{\epsilon}\right)$. Higher order integrators will keep the same scaling with $\lambda$ since with every order we bring down another factor of $\lambda$.

\subsection{Simulating Neutrino Oscillations using Hybrid Trotter-qDRIFT}
Now we will apply~\Cref{cor:hybridTrotter} to compare this cost to that required by the hybrid Trotter and continuous qDRIFT algorithm.   Specifically, the error in an $r$-segment simulation is of the form (under the assumption that $\mu \gg \omega$)
\begin{equation}
       \frac{t^2}{r} \bigg(c_I + 4 \sum_{k \neq l}^L \|H_k\|^2_{\infty} \bigg)\in O\left( \frac{N \mu^2t^2}{r} \right).
\end{equation}
As each segment of qubitization requires application of a first order Trotter formula, the cost per segment in terms of operator exponentials scales as $O(N^2)$.  Thus if we demand that the error is at most $\epsilon$, the cost is
\begin{equation}
    N_{\exp} \in O(N^2 r) \subseteq O\left( \frac{N^3\mu^2t^2}{\epsilon}\right),
\end{equation}
wherein each operator exponential requires $O(1)$ applications of $H, R_z$ and CNOT. This shows that the above asymptotic scaling applies in the gate complexity as well as the number of exponentials.

Interestingly, in the limit where $\lambda \gg 1$, this result provides better scaling than even the gate complexity of the truncated Dyson series~\cite{kieferova2019simulating,low2018hamiltonian} which scales in the interaction frame as $\log(\lambda)$.  In our case, the quantum computational complexity is completely independent of $\lambda$.  Of course, poly-logarithmic costs with these algorithms need to be incurred at the classical side to compute the rotation angles that go into the single qubit rotations but such costs are assumed to be negligible in our cost model. This implies that for such cases where the cost of the simulation is gated by the cost of preparing and controlling from the time-register, switching to a method that only requires classical controls can allow us to outperform such methods and make the gate count (rather than just the query complexity~\cite{low2018hamiltonian}) independent of the magnitude of the norm of the interaction Hamiltonian.

As a final note, similar scaling can also be attained by using the approach of~\cite{Poulin_2011} to time-order the operator exponentials that we use in the interaction frame.  The performance of this method is summarized in~\eqref{eq:c12def} and gives an alternative to the hybrid approach considered here and yields comparable scaling with $\lambda$. 

\section{Constrained Hamiltonian Dynamics}
\label{section:constraineddynamics}

As a final application of these techniques, let us consider the application of quantum simulation to dynamics subject to dynamical constraints.  Specifically, we will consider a Hamiltonian of the form
\begin{equation}
    H = H_{f} + \lambda P_c,
\end{equation}
where $H_f\in \mathbb{C}^{2^n\times 2^n}$ is the free Hamiltonian and $P_c \in \mathbb{C}^{2^n\times 2^n}$ is a projector onto an infeasible region.  The idea behind our approach to simulating constrained quantum dynamics is that if we choose $\lambda \gg \|H_f\|_{\infty}$ and an initial state $\ket{\psi} = (\openone -P_c) \ket{\psi}$, then the dynamics of the quantum system will, up to small errors, be confined to within the dynamically feasible region specified by the null-space of $P_c$.  Note that this result is reminiscent of others in the literature such as~\cite{oliveira2005complexity,cao2017efficient}; however, this result is specialized to time evolution and is simpler to employ in this context.

\begin{lemma}\label{lem:constraint}
Let $H_f \in \mathbb{C}^{2^n \times 2^n}$ be a free Hamiltonian for a system and let $\lambda$ be a variable describing the strength of the constraint such that $\|H_f\|_\infty \ll \lambda$.  We then have that for any $\ket{\psi}$ in the null-space of $P_c$, 
$$ \|e^{-i (H_f - \lambda P_c)t} \ket{\psi} - \lim_{\lambda \rightarrow \infty} e^{-i (H_f - \lambda P_c)t} \ket{\psi}  \|_2 \in O\left( \frac{\|H_f\|^2_{\infty} t}{\lambda}\right)$$ where $\|\cdot\|_2$ refers to the vector $2$-norm.
\end{lemma}
\begin{proof}
In order to show the deviation in each eigenvector of the Hamiltonian that arises from adding the small Hamiltonian $H_f$ to the constraint term, we will introduce
\begin{equation}
    H_f(x) := H_f x + \lambda P_c,
\end{equation}
where $x \in [0,1]$.  For $x=0,$ $H_f(0) = \lambda P_c$ has a degenerate null-space denoted $\mathcal{P}^0$.  Let $\ket{v_j(x)}$ denote the eigenvectors of $H_f(x)$ with corresponding eigenvalues $E_j(x)$.  
 We then have from perturbation theory that the derivative of $\ket{v_j(x)}$ is
\begin{equation}
    \frac{\partial  \ket{v_j(x)}}{\partial x} = \sum_{k\ne j} \ket{v_k(x)} \frac{\bra{v_k(x)} H_f \ket{v_j(x)}}{E_j(x) - E_k(x)}
\end{equation}
Assuming that the eigenvalue gaps are non-zero, we further have that the second derivative is finite. From the definition of a Riemann integral, we get
\begin{equation}
    \ket{v_j(1)} = \int_0^1\frac{\partial  \ket{v_j(x)}}{\partial x} \mathrm{d}x = \ket{v_j(0)} + \lim_{r\rightarrow \infty} \sum_{p=2}^r \sum_{k\ne j} \ket{v_k((p-1)/r)} \frac{\bra{v_k((p-1)/r)} H_f \ket{v_j((p-1)/r)}}{E_j((p-1)/r) - E_k((p-1)/r)} \frac{1}{r}.
\end{equation}
This expression allows us to relate the shift in the eigenvectors recursively.  First let us consider the initial time step.  As $H_f(x)$ is degenerate at $x=0,$ we can choose the eigenvectors such that the matrix with components $\bra{v_j(0)} H_f \ket{v_k(0)}$ is a diagonal matrix for all $\ket{v_j(0)},\ket{v_k(0)} \in \mathcal{P}^0$ or $\ket{v_j(0)},\ket{v_k(0)} \in {\mathcal{P}^0}^\perp$.  In the former case we have that $(\openone-P_c)\ket{v_j(0)} = \ket{v_j(0)}$, so $\bra{v_j(0)} H_f \ket{v_k(0)} = \bra{v_j(0)} (\openone - P_c) H_f (\openone - P_c) \ket{v_k(0)}$.  Thus we can achieve the diagonal criteria by choosing each $\ket{v_j(0)}$ to be an eigenvector of $(\openone - P_c) H_f (\openone - P_c)$.  Similarly, we can achieve the diagonal criteria for each vector in ${\mathcal{P}^0}^\perp$ by choosing each $\ket{v_j}$ to be an eigenvector of $P_c H_f P_c$.  We then have that for any $\ket{v_j(0)}$ in $\mathcal{P}^0$,
\begin{align}
    \ket{v_j(1/r)} &= \ket{v_j(0)}+\frac{1}{r}\sum_{k\ne j} \ket{v_k(0)} \frac{\bra{v_k(0)} H_f \ket{v_j(0)}}{E_j(0) - E_k(0)} \nonumber\\
    &=  \ket{v_j(0)}-\sum_{k:\ket{v_k(0)}\ne \mathcal{P}^0} \ket{v_k(0)} \frac{\bra{v_k(0)} H_f \ket{v_j(0)}}{\lambda r} \;.
\end{align}
In turn 
\begin{equation}
    \| \ket{v_j(0)} - \ket{v_j(1/r)} \|_2 \in O(\|H_f\|_\infty/\lambda r ).
\end{equation}
Furthermore from~\cite{horn}, we have that for all $k$, $|E_k(x) - E_k(0)| \le x \|H_f\|_\infty$.
Now let us assume that for some integer $q\ge 0$
\begin{equation}
    \| \ket{v_j(0)} - \ket{v_j(q/r)} \|_2 \in O(q\|H_f\|_\infty/\lambda r)\;.
\end{equation}
We then have 
\begin{equation}
\ket{v_j((q+1)/r)} = \ket{v_j(q/r)}+\frac{1}{r}\sum_{k\ne j} \ket{v_k(q/r)} \frac{\bra{v_k(q/r)} H_f \ket{v_j(q/r)}}{E_j(q/r) - E_k(q/r)} 
\end{equation}
and therefore
\begin{equation}
\|\ket{v_j((q+1)/r)} - \ket{v_j(q/r)}\|_2 \in O\left(\frac{ \|H_f\|_\infty}{(\lambda -2\|H_f\|_\infty)r}\right)=O\left(\frac{ \|H_f\|_\infty}{\lambda r}\right).
\end{equation}
Thus we have
\begin{align}
    \|\ket{v_j((q+1)/r)} - \ket{v_j(0)}\|_2 &\le \|\ket{v_j((q+1)/r)} - \ket{v_j(q/r)}\|_2 + \|\ket{v_j(0)} - \ket{v_j(q/r)}\|_2\nonumber\\
    &\in O\left(\frac{ (q+1)\|H_f\|_\infty}{\lambda r}\right) \;.
\end{align}
This in turn shows us that
\begin{equation}
    \|\ket{v_j(1)} - \ket{v_j(0)}\|_2 \in  O\left(\frac{ \|H_f\|_\infty}{\lambda }\right).
\end{equation}
Next by examining the differential equation for the eigenvalues, we have that the corresponding eigenvalue $E_j(1)$ obeys 
\begin{align}
    E_j(1) &= E_j(0)+\int_0^1 \frac{\partial E_j(x)}{\partial x} \mathrm{d} x = \int_0^1 \frac{\partial E_j(x)}{\partial x} \mathrm{d} x \nonumber\\
    &= \int_0^1 \bra{v_j(x)} H_f \ket{v_j(x)} \mathrm{d}x = \bra{v_j(0)} H_f \ket{v_j(0)} + O\left(\frac{ \|H_f\|_\infty^2}{\lambda }\right)\;.
\end{align}
Similarly for any $\ket{v_j(0)}\in {\mathcal{P}^0}^\perp$, $E_j(1) = \lambda +\bra{v_j(0)} H_f \ket{v_j(0)} + O\left(\frac{ \|H_f\|_\infty^2}{\lambda }\right)$.  We therefore have from the triangle inequality that
\begin{align}
    &\|H(1) - \sum_k (\lambda \delta_{{\ket{v_k} \in {\mathcal{P}^0}^\perp}} + \bra{v_k(0)} H_f \ket{v_k(0)} \ket{v_k(0)}\!\bra{v_k(0)})\|_\infty \nonumber\\
    &\in O\left(\frac{ \|H_f\|_\infty^2}{\lambda }\right)\;.
\end{align}
We therefore have from the fact that $\|e^{-iHt} - e^{-iH't}\|_\infty\le \|H-H'\|_\infty t$ for all Hermitian matrices $H$ and $H'$ of equal dimension that
\begin{equation}
    \|e^{-iH(1) t} - e^{-i \sum_k \left(\lambda \delta_{{\ket{v_k} \in {\mathcal{P}^0}^\perp}} + \bra{v_k(0)} H_f \ket{v_k(0)} \ket{v_k(0)}\!\bra{v_k(0)}\right) t}\|_\infty \in O\left(\frac{ \|H_f\|_\infty^2 t}{\lambda }\right).
\end{equation}
Therefore for any $\ket{\psi} \in \mathcal{P}^0$ we have that
\begin{equation}
    \|e^{-i H(1) t}\ket{\psi} - \lim_{\lambda \rightarrow \infty}e^{-i H(1) t}\ket{\psi}\|_2 \in O\left(\frac{ \|H_f\|_\infty^2 t}{\lambda }\right).
\end{equation}
\end{proof}
This shows that we can simulate constrained dynamics for time $t$ within error $\epsilon$ by choosing $\lambda \ge \|H_f\|^2_{\infty} t /\epsilon$.  This in turn leads to a substantial degradation of the scaling of most simulation algorithms if $[H_f,P_c]\ne 0$, because the Hamiltonian's norm scales with both the evolution time and the uncertainty desired in the simulation. This makes such constrained dynamics impractical for many applications.

This drawback can, however, be mitigated through the use of an interaction frame transformation. By transforming to the interaction frame of the constraint, we can perform the simulation at cost that is (in some cases) independent of the choice of $\lambda$.  The cost of such simulations using a hybrid qubitization and qDRIFT algorithm is given below.  We cite the complexity of this algorithm rather than truncated Dyson methods because such methods explicitly have a cost that scales logarithmically with $\lambda$; whereas in some cases the quantum gate complexity will be independent of $\lambda$.
\begin{theorem}
    Let the assumptions of~\Cref{thm:IPqDqubitsim} hold. Then there exists a quantum algorithm that implements, for any $t>0$ and $\epsilon>0$, a quantum channel that is a $(1,O(\log(L)), \epsilon)$ block encoding of $e^{-i(H_f +\lambda P_c)t}$.  Further this implementation requires  a total number of queries to PREPARE, SELECT and $W_{P_c}$ in
    $$
    O \bigg(\beta t + \left(\frac{\|H_f\|_\infty^2 t^2}{\epsilon}\right)\frac{\log(\|H_f\|_\infty t/\epsilon)}{\log \log(\|H_f\|_\infty t/\epsilon)}\bigg).
    $$
\end{theorem}
\begin{proof}
The proof follows directly from previous results.  
Specifically we have that
\begin{align}
    &\|V\ket{\psi} - \lim_{\lambda\rightarrow \infty}e^{-i(H_f + \lambda P_c)t}\|_2 \nonumber\\
    &\le \|V\ket{\psi} - e^{-i(H_f + \lambda P_c)t}\ket{\psi}\|_2 + \|e^{-i(H_f + \lambda P_c)t}\ket{\psi} - \lim_{\lambda\rightarrow \infty}e^{-i(H_f + \lambda P_c)t}\ket{\psi}\|_2 \;.
\end{align}
From~\autoref{thm:IPqDqubitsim}, we have that the number of queries to PREPARE, SELECT and $W_{P_c}$ needed to implement a $(1,O(\log(L)),\epsilon/2)$ block encoding is in
$$
O \bigg(\beta t + \left(\frac{\|H_f\|_\infty^2 t^2}{\epsilon}\right)\frac{\log(\|H_f\|_\infty t/\epsilon)}{\log \log(\|H_f\|_\infty t/\epsilon)}\bigg).
$$
Next, from~\autoref{lem:constraint} we have that there exists a value of $\lambda \in O(\|H_f\|_\infty^2 t/\epsilon)$ such that $\|e^{-i(H_f + \lambda P_c)t}\ket{\psi} - \lim_{\lambda\rightarrow \infty}e^{-i(H_f + \lambda P_c)t}\ket{\psi}\|_2\le \epsilon/2$.  The result then follows from the triangle inequality.
\end{proof}

These results show that query efficient methods exist for simulating Hamiltonian dynamics; however, the existence of a query efficient algorithm for simulating dynamics subject to a particular constraint does not imply the existence of a gate efficient algorithm.  For example, let us consider the case where $P_c \ket{x} = \ket{x}$ if and only if $E(x) \le \delta$ for some $\delta>0$ and $E(x)$ is the energy function for an arbitrary Ising model.  Since this problem is $\NP$-hard~\cite{barahona}, a gate  efficient version of this constraint is only possible if $\NP \subseteq \BQP$, which is strongly believed to be false.  For this reason, we provide below a sufficient, but not a necessary, condition for the $W_{P_c}$ to be simulatable in $O(\polylog(2^n\lambda t/\epsilon))$ gate operations.

\begin{proposition}
Let $P_c\in \mathbb{C}^{2^n \times 2^n}$ be a projector matrix. Suppose there exist functions such that for any $x,y \in \mathbb{Z}_{2^n}$, $g(x,y) = \bra{x} P_c \ket{y}$ and $f(x,i)$ yields the column index of the $i^{\rm th}$ non-zero matrix element of $P_c$ as represented in the computational basis. If
\begin{enumerate}
    \item $P_c$ has at most $1$ non-zero matrix elements in each row when expressed in the computational basis.
    \item $f$ is computable using a number of quantum gates that are in $O(\poly(n))$ and $g$ within error $2^{-m}$ using $O(\poly(nm))$ quantum operations
\end{enumerate}
then for any $\lambda \ge 0$ and $t\ge 0$, 
a unitary $\tilde{U}$ can be constructed such that $\|\tilde{U} - U_I(\lambda;t)\|_\infty \le \epsilon$ using $O(\polylog(2^n\lambda t/\epsilon))$ quantum operations.
\end{proposition}
\begin{proof}
The proof follows straight forwardly.  If we assume that $g(x,y)$ can be implemented within zero error with $m$ bits of precision, we have from~\cite{childs2003exponential} that $e^{-i\lambda  P_c t}$ can be implemented with zero error using $O(1)$ applications of $f$ and $g$ as well as $O({\rm poly}(nm))$ auxillary quantum operations.

Now let us assume that $g(x,y)$ cannot be computed within zero error using $m<\infty$ bits of precision.  If we denote $\tilde g(x,y)$ to be the approximate version of $g$, we have from the fact that $P_c$ is one-sparse that
\begin{align}
    & \|e^{-i\lambda t \sum_x g(x,f(x,1)) \ketbra{x}{f(x,1)}} - e^{-i\lambda t \sum_x \tilde g(x,f(x,1)) \ketbra{x}{f(x,1)}} \|_\infty \nonumber\\
    &= \max\Biggr(\max_{x: x\ne f(x,1)} \|e^{-i\lambda t g(x,f(x,1))( \ketbra{x}{f(x,1)}+\ketbra{f(x,1)}{x})} - e^{-i\lambda t \tilde g(x,f(x,1)) ( \ketbra{x}{f(x,1)}+\ketbra{f(x,1)}{x})} \|_\infty\nonumber\\
    &\qquad\qquad\qquad\qquad, \max_{x:x=f(x,1)} \|e^{-i\lambda t g(x,f(x,1))\ketbra{x}{x}} - e^{-i\lambda t \tilde g(x,f(x,1)) \ketbra{x}{x})} \|_\infty\Biggr)\nonumber\\
    &\le \lambda t \max_x |g(x,f(x,1)) - \tilde g(x,f(x,1))| \le \lambda t 2^{-m}\;.
\end{align}
Thus to achieve an error of $\epsilon$, we need to take $m \in O(\log(\lambda t /\epsilon))$.  The result immediately follows from the assumptions on the cost of $f$ and $g$.
\end{proof}

There are a number of applications of this approach to solve constrained versions of quantum dynamics.  One such application involves the simulation of quantum field theories within a particular gauge, which describes a choice of a dynamically unobservable feature of the system that is needed to unambiguously determine the dynamics.  For example, the Lorenz gauge involves choosing the vector potential such that $\partial_\mu A^\mu = 0$.  Rather than fixing the gauge by a clever choice of an equation of motion, this approach allows us to impose such gauges by penalizing all configurations that violate this.

Another application involves Gauss' law in quantum electrodynamics~\cite{Hauke2013,stryker2019oracles}. For $D$-dimensional quantum electrodynamics, Gauss' law reads $\nabla \cdot \hat{E}(s) - \hat{\rho}(s) =0$, where $\hat{E}(s)$ is the electric field operator at position $s$ and $\hat{\rho}(s)$ is the charge density there.  On a lattice, this can be further simplified to $G(s):=\sum_{i=1}^D (\hat{E}(s) - \hat{E}(s-e_i)) - \sum_{s,\sigma} e_\sigma n_\sigma(s) $, where $n_\sigma(s)$ is the number of electrons (or positrons) at a site and $e_\sigma$ is $\pm 1$ depending on the site.  From this, the constraint projector $P_c$ can be expressed using the properties of discrete Fourier transforms as $P_c = \openone -\frac{1}{N}\sum_{i=1}^N e^{-i 2\pi G(s)/N}$~\cite{stryker2019oracles}.  This is relevant because Gauss' law is only approximately held for methods such as Trotter-Suzuki simulations and so the application of this constraint oracle can be used to filter out the unphysical components of simulation error.  Specifically, consider a constraint on $\ket{\psi} \in \mathcal{P}^0$ given by the projector $P_c$ that commutes with the Hamiltonian. In other words, $[P_c, \sum_j H_j] =0$; however there may exist $k'$ such that $[H_{k'}, P_c]\ne 0$. This creates problems for the Trotter-Suzuki expansion, but we can address this by transforming into the interaction frame as discussed below
\begin{equation}
    e^{-i \sum_{j=1}^M H_j t} \ket{\psi} = e^{-i (\sum_{j=1}^M H_j + \lambda P_c) t} \ket{\psi} = e^{-i \lambda P_c t} \exp_{\tau}\left(\int_0^t e^{i\lambda P_c s} \sum_j H_j e^{-i\lambda P_c s} \mathrm{d}s \right)\ket{\psi}.~\label{eq:timeorder}
\end{equation}
We can then implement the time-ordered operator exponential in~\eqref{eq:timeorder} using one of our previous methods, such as a hybridized Trotter-qDRIFT method or that used in the previous section.  This allows us to impose a constraint, such as Gauss' law, on the integration formula at low cost. By contrast, if we were to try to do so using a high-order Trotter formula, we would have remainder terms in the Trotter-Suzuki expansion that are in $O(\poly(\lambda t))$. From~\autoref{lem:constraint}, this is in $O(\poly(\|H_f\|_\infty^2 t^2/\epsilon))$ and thus cannot be implemented at low cost in the limit where $\epsilon \ll 1$, unlike in the interaction picture approach.

A similar simple example of this is solving the Schrodinger equation for a particle constrained to a given surface.  As an example, consider solving the Schrodinger equation for a particle constrained to be on the surface of a Figure-8 immersion of a Klein bottle, which is a non-orientable surface with no boundary.  Such a surface is given, for some fixed value of $r>2$, by the following parameterized surface over the angles $\theta,v \in [0,2\pi)$
\begin{align}
    x&=(r+\cos(\theta/2) \sin(v) - \sin(\theta/2) \sin(2v))\cos(\theta)\nonumber\\
    y&=(r+\cos(\theta/2) \sin(v) - \sin(\theta/2) \sin(2v))\sin(\theta)\nonumber\\
    z&=\sin(\theta/2) \sin(v) + \cos(\theta/2) \sin(2v),
\end{align}
where in Cartesian coordinates $\theta = \arctan(y/x)$ and $v$ is found implicitly through the above expressions. As all these coordinate functions are Lipshitz continuous, a least square solution can be found using gradient descent after dividing up the surface in $(\theta,v)$ coordinates into a finite number of regions and then performing gradient descent of $\|\vec{x} - \vec{x}(r,\theta,v)\|$.  Thus by following this procedure, we can decide within $\epsilon$ error whether a given $(x,y,z)$ lies on the surface of Klein bottle.  In turn, $P_c$ can be constructed by using reversible logic to evaluate this in time $O(\poly(\log(1/\epsilon)))$. Thus, complicated quantum dynamics on unusual manifolds can be simulated through the use of our approach to constraints, even in cases like the figure-8 immersion of the Klein bottle where no simple coordinate system is available that makes the computation of the Laplacian operator in $(r,v,\theta)$ coordinates trivial. This is because the gradient fails to be defined there, as the normal vector cannot be defined at the intersection in the figure-8. Instead, we can rely on the constraint operator to force the dynamics to lie on the surface of the bottle and use the standard Laplacian in Cartesian coordinates.

As a final point of discussion, let us consider applying these ideas to simulate a universal Hamiltonian with a quantum circuit.  A universal Hamiltonian is a Hamiltonian such that the groundstate of the Hamiltonian encodes a quantum superposition of the history of a quantum computer via a clockstate of the form $\frac{1}{\sqrt{T}}\sum_t \ket{t} \ket{\psi(t)}$ where $\ket{\psi(t)}$ is the state of the quantum computer after $t$ gates have been applied to it~\cite{kempe20033,aharonov2008adiabatic,osborne2012hamiltonian,kohler2022translationally}.  In order to minimize the locality needed by these constructions, techniques such as ``perturbative gadgets'' are employed which allow restricted interactions such as $2$-local ones to simulate the action of a Hamiltonian of greater $k$-locality.  It is tempting therefore to ask whether our techniques could be used to accelerate the simulation of these constrained Hamiltonians.

As an example, the work of~\cite{kohler2022translationally} shows that a translationally invariant $1$D Hamiltonian of the following form for parameters $T$, $\Delta$ is universal
\begin{equation}
    H = \sum_{\langle i,j \rangle} \Delta h_{ij}^{(3)} + T\sum_i h_i^{(2)},
\end{equation}
where each $h_{ij}^{(3)}$ is a two-body translationally invariant Hamiltonian and each $h_i$ is a translationally invariant one-body Hamiltonian.  At first glance, the latter term appears to be fast forwardable.  This is significant, because the value $T$ corresponds to the number of gates employed in the circuit.  Thus we would be able to fast forward an arbitrary calculation if this were, by itself, true.

However, the value of $\Delta$ needed to provide a close approximation to the dynamics generically dominates the remaining term in~\cite{kohler2022translationally}. In particular if we demand a simulation error on the order of $\epsilon$, then it suffices to take $\Delta \in O(T^4/\epsilon)$ (for all other simulation parameters fixed).  Thus the translationally invariant $2$-body term dominates asymptotically and even if were possible to fast-forward the simulation of this Hamiltonian, the best case scenario would lead to a method that has scaling $O(T\log(1/\epsilon))$ from the one-body term. However, this construction is not self-evidently fast-forwardable and so a polynomial improvement is expected at best from transitioning to the interaction frame of the two-body operator.
\section{Conclusions}

We have developed novel simulation protocols that combine the standard simulation protocols of Trotterization, continuous qDRIFT, and qubitization in the interaction picture to simulate time-independent Hamiltonians. By exploiting the interaction picture, we can enter into the interaction frame of a fast-forwardable term with large or unbounded norm. Continuous qDRIFT is used to split the resulting time-ordered exponential into a product of time-independent exponentials with bounds proven for the number of time steps needed to achieve a desired error $\epsilon$. In the case of Hamiltonians with underlying commutator structure, Trotterizing first can reduce the query complexity further. Qubitization is then used for implementing the final time-independent exponentials, though other simulation techniques can be used. 

The hybrid protocol using Trotterization before continuous qDRIFT in the interaction frame of a fast-forwardable term in the Hamiltonian has a query complexity of $O(t^2(c_I + \sum_{k \neq l}^L \|H_k\|^2_{\infty})/\epsilon)$, where $c_I$ depends on the sum of norms of commutators. For Hamiltonian simulation problems with commutator structure, this is a drastic improvement over the complexity $O(\|H_k\|^2_{\infty,1,1}/\epsilon)$ obtained from directly employing conventional qDRIFT methods to a linear combination query model. The qubitization and continuous qDRIFT hybrid I.P. protocol has a query complexity bounded by $\tilde{O}(\lambda_{\alpha} t + \|H_{\alpha}\|^2_{\infty} t^2/\epsilon)$, where the quantities $\lambda_{\alpha}$ and $H_{\alpha}$ only involve the terms in the interaction Hamiltonian. If the term selected for the interaction frame is unbounded or of large operator norm, this again yields an improvement in the scaling with the $\ell^1$-norm of $H$ compared to qubitization.  Our approach does not require a complicated clock construction either, which makes it more practical than truncated Dyson series methods~\cite{low2018hamiltonian}.

Direct application of these techniques to the Schwinger Model yield a logarithmic scaling in the electric field cutoff $\Lambda$ for the query complexity. For the Hamiltonian model of collective neutrino oscillations, the query complexity is independent of the typically large constant $\lambda = \sqrt{2}G_f n_e$ representing the electron density, with the same scaling with respect to other parameters compared to conventional Trotter-Suzuki methods. The scaling with these parameters outperforms those achieved by current simulation methods. 

Further applications of these methods appear in simulating constrained dynamics. We show that the magnitude of the constraint term in the Hamiltonian needs to be prohibitively large to apply such a constraint using traditional simulation methods, such as qubitization.  However, using our approaches we can simulate the dynamics using a number of gate operations that (for certain constraints) is independent of the magnitude of the constraint.  This allows approximation methods similar to Trotter-Suzuki simulations to be employed while guaranteeing that the simulation does not break important symmetries present in the underlying dynamics (such as Gauss' law).

Another interesting fact to note is that even when Trotter formulas are used for the entire simulation, transforming to the interaction picture can have an advantage over performing the simulation in the laboratory frame.  This is because Trotter formulas have costs that scale with fractional powers of the derivatives and lead to costs that are linear in the strength of the interaction term, rather than a super-linear function as would be expected from a simulation in the laboratory frame~\cite{wiebe2010higher,childs2021theory}.  Although hybrid methods that provide $L^1$-norm scaling are shown to be advantageous in this regard, this advantage can be useful and may lead to improved methods to reduce the cost of simulation purely within the Trotter-Suzuki formalism wherein the structure of commutators can be more easily exploited.

These hybrid techniques are primarily useful in contexts where there are not only terms of large operator norm in a Hamiltonian but when those terms are diagonalizable, one-sparse, or more generally fast-forwardable. However, situations often arise in quantum simulation where it might be desirable to enter the interaction frame of terms that are not fast-forwardable, such as the hopping term $H_h$ of the Schwinger model in the continuum limit. As a na\"ive application of the present methods would involve doubling the number of times the non fast-forwardable term would be need to be simulated (see equation~\eqref{eq:expintHam}), additional work is needed to develop interaction picture algorithms, hybrid or otherwise, that are more optimized with respect to parameters that define certain physical regimes of interest. 

\begin{acknowledgements}
We thank Martin Savage for useful discussions. This work was supported in part by the U.S. Department of Energy, Office of Science, Office of Nuclear Physics, Inqubator for Quantum Simulation (IQuS) under Award Number DOE (NP) Award DE-SC0020970. It was further supported by a grant from Google research award, and NW's theoretical work on this project was  supported by the U.S. Department of Energy, Office of Science, National Quantum Information Science Research Centers, Co-Design Center for Quantum Advantage under contract number DE-SC0012704.
\end{acknowledgements}

\bibliographystyle{unsrtnat}
\bibliography{ipref}

\appendix

\section{Diamond Norm}
\label{section:diamond}
The diamond distance is often used as a measure of error between two quantum channels. It is defined as follows: 
%%%%%%%%
\begin{equation}
d_{\diamond}(\mathcal{E},\mathcal{N}) = \frac{1}{2}||\mathcal{E} - \mathcal{N}||_{\diamond},
\end{equation} 
%%%%%%%%%
where $||\ldots||_{\diamond}$ is the diamond norm
%%%%%%%%%
\begin{equation}
||\mathcal{P}||_{\diamond} \coloneqq \mathrm{sup}_{\rho; ||\rho||_1=1} ||(\mathcal{P} \otimes I)(\rho)||_1
\end{equation} 
%%%%%%%%%
and $\mathcal{E}$ and $\mathcal{N}$ are two quantum channels or superoperators. Note that $I$ acts on the same size Hilbert space as $\mathcal{P}$ and $\rho$ is a density matrix. All operators here are expressed as square matrices and $\rho$ can represent states entangled with qubits that are not operated on. We then need an identity matrix to ``pad out" the missing dimensions so that $\mathcal{P} \otimes I$ can act sensibly upon $\rho$. 

The diamond norm is simply the trace distance but maximized over all possible input states and satisfies two key properties:
%%%%%%%%%
\begin{enumerate}[(1)]

\item Triangle inequality: $||\mathcal{A} + \mathcal{B}||_{\diamond} \leq ||A||_{\diamond} + ||\mathcal{B}||_{\diamond}$

\item Sub-multiplicativity: $||\mathcal{A} \mathcal{B}||_{\diamond} \leq ||\mathcal{A}||_{\diamond} ||\mathcal{B}||_{\diamond}$ 

\end{enumerate}

It follows from the definition of the diamond norm that if we apply the channel $\mathcal{E}$ and $\mathcal{N}$ to the quantum state $\sigma$, we have
%%%%%%%%%%
\begin{equation}
d_{tr}(\mathcal{E}(\sigma), \mathcal{N}(\sigma)) = \frac{1}{2}||\mathcal{E}(\sigma) - \mathcal{N}(\sigma)||_1 \leq d_{\diamond}(\mathcal{E}, \mathcal{N}). \label{eq:tracediamond}
\end{equation}
%%%%%%%%

The trace norm distance is important since it bounds the error in expectation values. To see this, consider the expression $|\text{Tr}(M \mathcal{E}(\sigma)) - \text{Tr}(M \mathcal{N}(\sigma))|$. The expectation value of an operator $M$ with respect to a state $\rho$ can be found by taking the trace of their product, i.e $\text{Tr}(M\rho)$. Thus, in the above expression, we first send a state $\sigma$ through our two channels. Then we find the expectation value of $M$ with respect to their outputs and take the absolute value of the difference to find the error in expectation values. 

We can bound this error in expectation values by the following inequalities 
%%%%%%%%%%
\begin{align}
|\text{Tr}(M \mathcal{E}(\sigma)) - \text{Tr}(M \mathcal{N}(\sigma))| &= |\text{Tr}[M(\mathcal{E}(\sigma) - \mathcal{N}(\sigma))]| \leq 2 ||M|| d_{tr}(\mathcal{E}(\sigma), \mathcal{N}(\sigma)) \nonumber \\ 
&\leq 2 ||M|| d_{\diamond}(\mathcal{E},\mathcal{N}).
\end{align}
%%%%%%

In the first inequality, the von-Neumann trace inequality $$|\text{Tr}(AB)| \leq \sum_{i=1}^n \alpha_i \beta_i$$ was used where $\alpha_i, \beta_i$ are the singular values of the operators $A$ and $B$ respectively. This inequality can be further bounded by recognizing that $\alpha_i \leq \alpha_{\text{max}}$ for all $i$, where $\alpha_{\text{max}}$ is the largest singular value of $A$. The largest singular value of $A$ is precisely $||A||_{\infty}$ so $$|\text{Tr}(AB)| \leq \sum_{i=1}^n \alpha_i \beta_i \leq \sum_{i=1}^n \alpha_{\text{max}} \beta_i = ||A||_{\infty} \ ||B||_1.$$ 

The second inequality in the above expression follows directly from the definition of the diamond norm. 

Now note that if we have a projection operator $P$, $P^{\dag} P = P$ since projection operators are Hermitian and square to themselves. Their eigenvalues are $1$ and $0$ so it immediately follows that $||P||_{\infty} = 1$. So if $M$ is a projection operator and we have $\varepsilon$ error in the diamond distance, then $$|\text{Tr}(M \mathcal{E}(\sigma)) - \text{Tr}(M \mathcal{N}(\sigma))| \leq 2\varepsilon.$$ 

This justifies the statement that measurement statistics are correct up to a factor of $2\varepsilon$ with an $\varepsilon$ error in diamond distance. 

\section{Notation for qDRIFT}
\label{section:notation}

We establish the following notational conventions from \cite{berry2020time} for describing the time-dependent qDRIFT scaling. Let $\alpha \in \bbC^L$ be a vector. The notation $||\alpha||_p$ represents the $l_p$ norm of $\alpha$ and we define a few cases as follows:

$$||\alpha||_1 \coloneqq \sum_{j=1}^L |\alpha_j|, \hspace{0.5cm} ||\alpha||_2 \coloneqq \sqrt{\sum_{j=1}^L |\alpha_j|^2}, \hspace{0.5cm} ||\alpha||_{\infty} \coloneqq \max\limits_{j \in \{1,2,\ldots,L\}} |\alpha_j|\;.$$ %%%%%% 
%%%%%%

If $A$ is a matrix, $||A||_p$ denotes the Schatten-$p$ norm of $A$. A few important examples are: $$||A||_1 \coloneqq \text{Tr}(\sqrt{A^{\dag}A}), \hspace{0.5cm} ||A||_2 \coloneqq \sqrt{\text{Tr}(A^{\dag}A)}, \hspace{0.5cm} ||A||_{\infty} \coloneqq \max_{\ket{\psi}} ||A \ket{\psi}||_2\;.$$ %%%%%%%
%%%%%%%

If $f: [0, t] \rightarrow \bbC$ is a continuous function, $||f||_p$ denotes the $L^p$ norm of the function. Thus, $$||f||_1 \coloneqq \int_0^t d\tau |f(\tau)|, \hspace{0.5cm} ||f||_2 \coloneqq \sqrt{\int_0^t d\tau |f(\tau)|^2}, \hspace{0.5cm} ||f||_{\infty} \coloneqq \max\limits_{\tau \in [0,t]} |f(\tau)|\;.$$%%%%%%% 
%%%%%%%%

These norms can be combined to obtain vector and operator-valued functions. Suppose $\alpha: [0,t] \rightarrow \bbC^L$ is a continuous vector-valued function with components at time $\tau$ denoted by $\alpha_j(\tau)$. $||\alpha||_{p,q}$ denotes taking the $l_p$ norm $||\alpha(\tau)||_p$ for all $\tau$ and computing the $L^q$ norm of the resulting scalar function, e.g. $$||\alpha||_{1,1} = \int_0^t d\tau \sum_{j=1}^L |\alpha_j|\;.$$%%%%%%%
%%%%%%%%

Similar reasoning applies when dealing with the Schatten $p$-norm of a time-dependent operator and then applying an $L^q$ norm to the resulting scalar function, i.e. $||A||_{p,q}$. For example,  $$||A||_{1,2} = \sqrt{\int_0^t d\tau \ (\text{Tr}\big(\sqrt{A^{\dag}A}\big))^2}\;.$$
Note that this notation, while compact, is not well suited for describing evolution within a sub-interval of the entire evolution.  In the event that a shorter evolution needs to be described, we explicitly use the integral expression over the domain in question.

For time-dependent linear combinations $A(\tau) = \sum_{l=1}^L A_l(\tau)$, the notation $||A||_{p,q,r}$ means taking the Schatten $p$-norm $||A_t(\tau)||_p$ of each term in the sum and applying the $l_q$ and $L^r$ norms to the resulting vector-valued functions, e.g. $$\|A\|_{1,1,\infty} \coloneqq \max_{\tau \in [0,t]} \sum_{l=1}^L \|A_l(\tau) \|_1\;.$$

\end{document}